\documentclass[letter,twocolumn]{article} 

    \usepackage[hmargin=1.6cm,vmargin=2.7cm,centering]{geometry}
    \usepackage[utf8]{inputenc}
    \usepackage[OT4]{fontenc}

    \usepackage{latexsym,stmaryrd,amsmath,amssymb,amsthm,amsfonts,xspace}
    \usepackage{enumerate,verbatim,paralist}
    \usepackage{macros}
    \usepackage{url}

    \title{Reducing Protocol Analysis with XOR to the XOR-free Case in
           the\\ Horn Theory Based Approach\footnote{An
             abridged version of this paper appears in CCS
             2008 \cite{KuestersTruderung-CCS-2008}. This work was partially supported
             by the DFG under Grant KU 1434/4-1, the SNF
             under Grant 200021-116596, and the Polish Ministry of Science
             and Education under Grant 3 T11C 042 30.}}
    \author{ Ralf K\"usters\\{\normalsize University of Trier, Germany}\\\texttt{kuesters@uni-trier.de}
    \and Tomasz Truderung\\{\normalsize University of
      Trier, Germany and}\\
    {\normalsize Wrocław University, Poland}\\\texttt{truderun@uni-trier.de} }
    \date{}
    \newenvironment{ienum}{\begin{compactenum}[(i)]}{\end{compactenum}}
    \newenvironment{aenum}{\begin{compactenum}[(a)]}{\end{compactenum}}
    \newenvironment{nenum}{\begin{compactenum}[1.]}{\end{compactenum}}
    \def\parag#1{\paragraph{#1.}}

\begin{document}

    \maketitle

    \begin{abstract}\small
      In the Horn theory based approach for cryptographic
      protocol analysis, cryptographic protocols and
      (Dolev-Yao) intruders are modeled by Horn theories and
      security analysis boils down to solving the derivation
      problem for Horn theories.  This approach and the tools
      based on this approach, including ProVerif, have been
      very successful in the automatic analysis of
      cryptographic protocols w.r.t.~an unbounded number of
      sessions. However, dealing with the algebraic properties
      of operators such as the exclusive OR (XOR) has been
      problematic. In particular, ProVerif cannot deal with
      XOR.

      In this paper, we show how to reduce the derivation
      problem for Horn theories with XOR to the XOR-free case.
      Our reduction works for an expressive class of Horn
      theories. A large class of intruder capabilities and
      protocols that employ the XOR operator can be modeled by
      these theories. Our reduction allows us to carry out
      protocol analysis by tools, such as ProVerif, that cannot
      deal with XOR, but are very efficient in the XOR-free
      case. We implemented our reduction and, in combination
      with ProVerif, applied it in the automatic analysis of
      several protocols that use the XOR operator.  In one
      case, we found a new attack.
    \end{abstract}

\section{Introduction}

In the Horn theory based approach for cryptographic
protocol analysis, cryptographic protocols and the
so-called Do\-lev-Yao intruder are modeled by Horn
theories.  The security analysis, including the analysis of
secrecy and authentication properties, then essentially
boils down to solving the derivation problem for Horn
theories, i.e., the question whether a certain fact is
derivable from the Horn theory. This kind of analysis takes
into account that an unbounded number of protocol sessions
may run concurrently. While the derivation problem is
undecidable in general, there are very successful automatic
analysis tools, with ProVerif \cite{Blanchet-CSFW-2001}
being one of the most promintent ones among them, which
work well in practice.

However, dealing with the algebraic properties of
operators, such as the exclusive OR (XOR), which are
frequently used in cryptographic protocols, has been
problematic in the Horn theory approach. While ProVerif has
been extended to deal with certain algebraic properties in
\cite{BlanchetAbadiFournet-JLAP2007}, associative
operators, which in particular include XOR, are still out
of the scope. Even though there exist some decidability
results for the derivation problem in certain classes of
Horn theories with XOR
\cite{ComonCortier-RTA-2003, VermaSeidlSchwentick-CADE-2005,
CortierKeighrenSteel-TACAS-2007}, the
decision procedures have not led to practical implementations yet,
except for the very specific setting in
\cite{CortierKeighrenSteel-TACAS-2007} (see the related work).

The goal of this work is therefore to come up with a
practical approach that allows for the automatic analysis
of a wide range of cryptographic protocols with XOR, in a
setting with an unbounded number of protocol sessions.  Our
approach is to reduce this problem to the one without XOR,
i.e., to the simpler case without algebraic properties.
This simpler problem can then be solved by tools, such as
ProVerif, that a priori cannot deal with XOR, but are very
efficient in solving the XOR-free case. More precisely, the
contribution of this paper is as follows.

\parag{Contribution of this paper}
We consider an expressive class of (unary) Horn theories,
called $\xor$-linear (see Section~\ref{sec:dominated}). A
Horn theory is $\xor$-linear, if for every Horn clause in
this theory, except for the clause that models the
intruder's ability to apply the XOR operator
($I(x),I(y)\rightarrow I(x\xor y)$), the terms that occur
in these clauses are $\xor$-linear. A term is $\xor$-linear
if for every subterm of the form $t\xor t'$ in this term,
it is true that $t$ or $t'$ does not contain variables. We
do not put any other restriction on the Horn theories. In
particular, our approach will allow us to deal with all
cryptographic protocols and intruder capabilities that can
be modeled as $\xor$-linear Horn theories.

We show that the derivation problem for $\xor$-linear Horn
theories with XOR can be reduced to a purely syntactic
derivation problem, i.e., a derivation problem where the
algebraic properties of XOR do not have to be considered
anymore (see Section~\ref{sec:dominated},
\ref{sec:reduction}, and \ref{sec:authentication}). Now,
the syntactic derivation problem can be solved by highly
efficient tools, such as ProVerif, which cannot deal with
XOR.  We believe that the techniques developed in this
paper are interesting beyond the case of XOR. For example,
using these techniques it might be possible to also deal
with other operators, such as
Diffie-Hellman-Exponentiation.

Using ProVerif, we apply our two step approach---first reduce the
problem, then run ProVerif on the result of the reduction---to the
analysis of several cryptographic protocols that use the XOR operator
in an essential way (see
Section~\ref{sec:implementationandexperiments}).  The experimental
results demonstrate that our approach is practical.  In one case, we
found a new attack on a protocol.

We note that a potential alternative to our approach is to perform
unification modulo XOR instead of syntactic unification in a
resolution algorithm such as the one employed by ProVerif. Whether or
not this approach is practical is an open problem. The main difficulty
is that unification modulo XOR is much more inefficient than syntactic
unification; it is NP-complete rather than linear and, in general,
there does not exist a (single) most general unifier.

\parag{Related work}
In \cite{ComonCortier-RTA-2003,VermaSeidlSchwentick-CADE-2005},
classes of Horn theories (security protocols) are
identified for which the derivation problem modulo XOR is
shown to be decidable.  These classes are orthogonal to the
one studied in this paper. While $\xor$-linearity is not
required, other restrictions are put on the Horn clauses,
in particular linearity on the occurrence of variables. The
classes in
\cite{ComonCortier-RTA-2003,VermaSeidlSchwentick-CADE-2005}
do, for example, not contain the Recursive Authentication
and the SK3 protocol, which, however, we can model (see
Section~\ref{sec:implementationandexperiments}). To the
best of our knowledge, the decision procedures proposed in
\cite{ComonCortier-RTA-2003,VermaSeidlSchwentick-CADE-2005}
have not been implemented. The procedure proposed in
\cite{ComonCortier-RTA-2003} has non-elementary runtime.

In \cite{Steel-CADE-2005, CortierKeighrenSteel-TACAS-2007,
  CortierDelauneSteel-CSF-2007}, the IBM 4758 CCA API,
which we also consider in our experiments, has been
analyzed. Notably, in
\cite{CortierKeighrenSteel-TACAS-2007} a decision
procedure, along with an implementation, is presented for
the automatic analysis of a class of security protocols
which contains the IBM 4758 CCA API.  However, the protocol
class and the decision procedure is especially tailored to
the IBM 4758 CCA API. The only primitives that can be
handled are the XOR operator and symmetric encryption. All
other primitives, such as pairing, public-key encryption,
and hashing, are out of the scope of the method in
\cite{CortierKeighrenSteel-TACAS-2007}. The specification
of the IBM 4758 CCA API in
\cite{CortierKeighrenSteel-TACAS-2007} is hard coded in a C
implementation.

In \cite{BlanchetAbadiFournet-JLAP2007}, it is described
how the basic resolution algorithm used in ProVerif
can be extended to handle some equational theories.
However, as already mentioned in that work, associative
operators, such as XOR, are out of the scope of this
extension.

In \cite{ComonDelaune-RTA-2005}, the so-called finite
variant property has been studied for XOR and other
operators.  It has been used (implicitly or explicitly) in
other works
\cite{ComonShmatikov-LICS-2003,ComonCortier-RTA-2003}, and
also plays a role in our work (see
Section~\ref{sec:reduction}).

In
\cite{ChevalierKuestersRusinowitchTuruani-LICS-2003,ComonShmatikov-LICS-2003,KuestersTruderung-STACS-2007},
decision procedures for protocol analysis with XOR w.r.t.~a
\emph{bounded} (rather than an unbounded) number of
sessions are presented. The notion of $\xor$-linearity that
we use is taken from the work in
\cite{KuestersTruderung-STACS-2007}. That work also
contains some reduction argument. However, our work is
different to \cite{KuestersTruderung-STACS-2007} in several
respects: First, of course, our approach is for an
\emph{unbounded} number of sessions, but it is not
guaranteed to terminate. Second, the class of protocols
(and intruder capabilities) we can model in our setting is
much more general than the one in
\cite{KuestersTruderung-STACS-2007}.  Third, the reduction
presented in \cite{KuestersTruderung-STACS-2007} heavily
depends on the bounded session assumption; the argument
would not work in our setting. Fourth, the reduction
presented in \cite{KuestersTruderung-STACS-2007} is not
practical.

\parag{Structure of this paper} In
Section~\ref{sec:preliminaries}, we introduce Horn theories
and illustrate how they are used to model cryptographic
protocols by a running example. The notion of
$\xor$-linearity is introduced in
Section~\ref{sec:dominated}, along with a proposition that
is the key to our main result, i.e., the reduction.  The
reduction is then presented in Section~\ref{sec:reduction},
with extensions to authentication presented in
Section~\ref{sec:authentication}. We discuss our
implementation and experimental results in
Section~\ref{sec:implementationandexperiments}. Proofs
omitted in the main part of the paper are presented in the
appendix.

We point the reader to \cite{implementation} for our implementation.


    \section{Preliminaries}\label{sec:preliminaries}

In this section, we introduce Horn theories modulo the XOR
operator and illustrate how these theories are used to
model the so-called Dolev-Yao intruder and cryptographic
protocols by a running example.

\subsection*{Horn theories}
Let $\Sigma$ be a finite signature and $V$ be a set of
variables.  The set of terms over $\Sigma$ and $V$ is
defined as usual.  By $\var(t)$ we denote the set of
variables that occur in the term $t$.  We assume $\Sigma$
to contain the binary function symbol $\xor$
(\emph{exclusive OR}), as well as a constant $0$. To model
cryptographic protocols, $\Sigma$ typically also contains
constants (\emph{atomic messages}), such as principal
names, nonces, and keys, the unary function symbol
$\hash(\cdot)$ (\emph{hashing}), the unary function symbol
$\pub{\cdot}$ (\emph{public key}), and binary function
symbols such as $\an{\cdot,\cdot}$ (\emph{pairing}),
$\enc{\text{\normalsize$\cdot$}}\cdot$ (\emph{symmetric
  encryption}), and
$\penc{\text{\normalsize$\cdot$}}{\!\cdot\!}$ (\emph{public
  key encryption}). The signature $\Sigma$ may also contain
any other free function symbol, such as various kinds of
signatures and MACs. We only require that the corresponding
intruder rules are $\xor$-linear (see Section
\ref{sec:dominated}), which rules that do not contain the
symbol $\xor$ always are.

\emph{Ground terms}, i.e.\ terms without variables, are
called \emph{messages}.  For a unary predicate $q$ and a
(ground) term $t$ we call $q(t)$ a \emph{(ground) atom}.  A
\emph{substitution} is a finite set of pairs of the form 
$\sigma = \{t_1/x_1,\dots,t_n/x_n\}$, where $t_1,\dots,t_n$ are terms
and $x_1,\dots, x_n$ are variables. The set
$\dom(\sigma)=\{x_1,\dots,x_n\}$ is called the domain of
$\sigma$. We define $\sigma(x)=x$ if $x\notin
\dom(\sigma)$. The application $t\sigma$ of $\sigma$ to a
term/atom/set of terms $t$ is defined as usual.

We call a term \emph{standard} if its top-symbol is not
$\xor$; otherwise, it is called \emph{non-standard}. For
example, the term $\an{a,b\xor a}$ is standard, while
$b\xor a$ is non-standard. 

A non-standard subterm $s$ of $t$ is called
\emph{complete}, if either $s=t$ or $s$ occurs in $t$ as a
direct subterm of some standard term.  For instance, for $t
= \an{a \xor \senc{(x\xor y) \xor z}y,b}$, the terms $a
\xor \senc{(x\xor y) \xor z}{y}$ and $(x\xor y)\xor z$ are
complete non-standard subterms of $t$, but $x\xor y$ is
not.

    
To model the algebraic properties of the exclusive OR
(XOR), we consider the congruence relation $\sim$ on terms
induced by the following equational theory (see, e.g.,
\cite{ComonShmatikov-LICS-2003,ChevalierKuestersRusinowitchTuruani-LICS-2003}):
\begin{align} 
\label{eq:xor1}  
    x\xor y&=y\xor x &
    (x\xor y)\xor z&=x\xor (y\xor z)\\
\label{eq:xor2}  
    x\xor x&=0
    &x\xor 0&=x
\end{align} 
For example, we have that $t_{ex}=a\xor b\xor \enc{k}{0}
\xor b\xor\enc{k}{c\xor c}\sim a$.  (Due to the
associativity of $\xor$ we often omit brackets and simply
write $a\xor b\xor c$ instead of $(a\xor b)\xor c$ or
$a\xor (b\xor c)$.) For atoms $q(t)$ and $q'(t')$, we write
$q(t)\sim q'(t')$ if $q=q'$ and $t\sim t'$.  We say that
two terms are \emph{equivalent modulo AC}, where AC stands
for associativity and commutativity, if they are equivalent modulo
\eqref{eq:xor1}. A term is \emph{$\xor$-reduced} if modulo AC, the
identities \eqref{eq:xor2}, when interpreted as reductions from
left to right, cannot be applied. Clearly, every term can be turned
into $\xor$-reduced form and this form is uniquely determined modulo
AC. For example, $a$ is the $\xor$-reduced form of $t_{ex}$.
    
A \emph{Horn theory} $T$ is a finite set of \emph{Horn clauses} of the
form $a_1,\ldots,a_n\ra a_0$, where $a_i$ is an atom for every $i\in
\{0,\ldots,n\}$. We assume that the variables that occur on the
right-hand side of a Horn clause also occur on the left-hand
side\footnote{This assumption can easily be relaxed for variables that
are substituted only be cetrain ``good'' terms, where ``good'' means
$\C$-dominated (see Section \ref{sec:dominated})}.  If $n=0$, i.e.,
the left-hand side of the clause is always true, we call
the Horn clause $a_0$ a \emph{fact}.
    
Given a Horn theory $T$ and a ground atom $a$, we say that \emph{$a$
can syntactically be derived from $A$ w.r.t.~$T$ } (written $T\vdash
a$) if there exists a \emph{derivation} for $a$ from $T$,
i.e., there exists a sequence $\pi=b_1,\ldots,b_l$ of ground atoms
such that $b_l = a$ and for every $i\in \{1,\ldots,l\}$ there exists a
substitution $\sigma$ and a Horn clause $a_1,\ldots,a_n\rightarrow
a_0$ in $T$ such that $a_0\sigma= b_i$ and for every $j\in
\{1,\ldots,n\}$ there exists $k\in \{1,\ldots,i-1\}$ with $a_j\sigma =
b_k$. In what follows, we sometimes refer to $b_i$ by $\pi(i)$ and to
$b_1,\ldots,b_i$ by $\pi_{\le i}$. The \emph{length} $l$ of a
derivation $\pi$ is referred to by $|\pi|$.

We call a sequence $b_1,\ldots,b_l$ of ground atoms an
\emph{incomplete syntactic derivation of $a$ from $T$} if
$b_l=a$ and $T\cup \{b_1,\ldots,b_{i-1}\}\vdash b_i$ for
every $i\in \{1,\ldots,b_l\}$.

Similarly, we write $T\svdash a$ if there exists a
\emph{derivation of $a$ from $T$ modulo XOR}, i.e., there
exists a sequence $b_1,\ldots,b_l$ of ground atoms such
that $b_l \sim a$ and for every $i\in \{1,\ldots,l\}$ there
exists a substitution $\sigma$ and a Horn clause
$a_1,\ldots,a_n\rightarrow a_0$ in $T$ such that
$a_0\sigma\sim b_i$ and for every $j\in \{1,\ldots,n\}$
there exists $k\in \{1,\ldots,i-1\}$ with $a_j\sigma\sim
b_k$. \emph{Incomplete derivations modulo XOR} are defined
analogously to the syntactic case.

Given $T$ and $a$, we call the problem of deciding whether $T\vdash a$
($T\svdash a$) is true, the \emph{deduction problem (modulo XOR)}. In
case $T$ models a protocol and the intruder (as described below), the
fact that $T\svdash a$, with $a=\I(t)$, is \emph{not} true means that
the term $t$ is secret, i.e., the intruder cannot get hold of $t$ even
when running an unbounded number of sessions of the protocol and using
algebraic properties of the XOR operator.


\subsection*{Modeling Protocols by Horn theories}
Following \cite{Blanchet-CSFW-2001}, we now illustrate how
Horn theories can be used to analyze cryptographic
protocols, where, however, we take the XOR operator into
account. While here we concentrate on secrecy properties,
authentication is discussed in
Section~\ref{sec:authentication}. As mentioned in the
introduction, the Horn theory approach allows us to analyze
the security of protocols w.r.t.~an unbounded number of
sessions and with no bound on the message size in a fully
automatic and sound way.  However, the algorithms are not
guaranteed to terminate and may produce false attacks.

\begin{figure}
    \small
        \begin{align*} 
            \I(x) &\ra \I(\hash(x)) 
                & \I(x), \I(y) &\ra \I(\an{x,y}) \\[0.3ex]
            \I(\an{x,y}) &\ra \I(x) 
                & \I(\an{x,y}) &\ra \I(y)\\[0.3ex]
            \I(x), \I(y) &\ra \I(\senc xy), 
                & \I(\senc xy), \I(y) &\ra \I(x) \\[0.3ex]
            \I(x), \I(\pub y) &\ra \I(\penc x{\pub y}), 
                & \I(\penc x{\pub y}), \I(y) &\ra \I(x) \\[0.3ex]
            \I(x), \I(y) &\ra \I(x\xor y)
        \end{align*}
        \caption{Intruder Rules.\label{intruder-rules}}
\end{figure}

A Horn theory for modeling protocols and the (Dolev-Yao)
intruder uses only the predicate $\I$.  The fact $\I(t)$
means that the intruder may be able to obtain the term $t$.
The fundamental property is that if $\I(t)$ cannot be
derived from the set of clauses, then the protocol
preserves the secrecy of $t$.  The Horn theory consists of
three sets of Horn clauses: the initial intruder facts, the
intruder rules, and the protocol rules. The set of
\emph{initial intruder facts} represents the initial
intruder knowledge, such as names of principals and public
keys.  The clauses in this set are facts, e.g., $\I(a)$
(the intruder knows the name $a$) and $\I(\pub{sk_a})$ (the
intruder knows the public key of $a$, with $sk_a$ being the
corresponding private key).  The set of \emph{intruder
  rules} represents the intruders ability to derive new
messages. For the cryptographic primitives mentioned above,
the set of intruder rules consists of the clauses depicted
in Figure~\ref{intruder-rules}. The last clause in this
figure will be called the \emph{$\xor$-rule}. It allows the
intruder to perform the XOR operation on arbitrary
messages. The set of \emph{protocol rules} represents the
actions performed in the actual protocol. The $i$th
protocol step of a principal is described by a clause of
the form $\I(r_1),\ldots,\I(r_i)\rightarrow \I(s_i)$ where
the terms $r_j$, $j\in \{1,\ldots,i\}$, describe the
(patterns of) messages the principal has received in the
previous $i{-}1$st steps plus the (pattern of the) message in
the $i$th step.  The term $\I(s_i)$ is the (pattern of) the
$i$th output message of the principal.  Given a protocol
$P$, we denote by $T_P$ the Horn theory that comprises all
three sets mentioned above.

Let us illustrate the above by a simple example protocol,
which we will use as a running example throughout this
paper. Applications of our approach to more complex
protocols are presented in Section~\ref{sec:experiments}.
We emphasize that the kind of Horn theories outlined above
are only an example of how protocols and intruders can be
modeled. As already mentioned in the introduction, our
methods applies to all $\xor$-linear Horn theories.

\subsection*{Running example} 
We consider a protocol that was proposed in
\cite{ChevalierKuestersRusinowitchTuruani-LICS-2003}. It is a variant
of the Needham-Schroeder-Lowe protocol in which XOR is employed. The
informal description of the protocol, which we denote by
$P_{\mathit{NSL}_{\xor}}$, is as follows:
\begin{protocol}
        \ab{$\penc{\an{N,A}}{\pub{sk_B}}$}
        \ba{$\penc{\an{M, N\xor B}}{\pub{sk_A}}$}
        \ab{$ \penc{M}{\pub{sk_B}}$}
\end{protocol}
where $N$ and $M$ are nonces generated by $A$ and $B$,
respectively. As noted in
\cite{ChevalierKuestersRusinowitchTuruani-LICS-2003}, this
protocol is insecure; a similar attack as the one on the
original Needham-Schroeder protocol can be mounted, where,
however, now the algebraic properties of XOR are
exploited.

To illustrate how this protocol can be modeled in terms of
Horn theories, let $\part$ be a set of participant names
and $\honest\subseteq\part$ be the set of names of the
honest participants. As proved in
\cite{ComonLundhCortier-SCP-2004}, for the secrecy property
it suffices to consider the case $\part = \{a,b\}$ and
$\honest = \{a\}$ (for authentication three participants
are needed).  In the following, $sk_a$, for $a\in\part$,
denotes the private key of $a$, $n(a,b)$ denotes the nonce
sent by $a\in\part$ to $b\in\part$ in message 1., and
$m(b,a)$ denotes the nonce generated by $b$ and sent to $a$
in message 2.

The initial intruder knowledge is the set
$
\{\I(a)\mid a\in \part\}\cup
\{\I(\pub{sk_a})\mid a\in \part\}\cup \{\I(sk_a)\mid a\in
\part\setminus \honest\}
$ 
of facts. The intruder rules are those depicted in
Figure~\ref{intruder-rules}.
The first step of the protocol performed by an honest
principal is modeled by the facts: 
$$
        \I(\penc{\an{n(a,b),a}}{\pub{sk_b}})
$$
for $a\in\honest$, $b\in\part$.  Note that it is not
necessary to model messages sent by dishonest principals,
since these are taken care of by the actions that can be
performed by the intruder.
    
The second step of the protocol performed by an honest
principal is modeled by the clauses: 
\begin{equation} \label{ckrt-fragile} 
\I(\penc{\an{x,
      a}}{\pub{sk_b}}) \ra \I(\penc{\an{m(b,a), x\xor
    b}}{\pub{sk_a}})
\end{equation}
for $b\in\honest$, $a\in\part$. The third step of the
protocol performed by an honest principal is modeled by the
clauses: 
\begin{equation} \label{ckrt-last}
        \I(\penc{\an{y, n(a,b)\xor b}}{\pub{sk_a}}) \ra \I(\penc{y}{\pub{sk_b}})
\end{equation}
for $a\in\honest$, $b\in\part$. The set of Horn clauses
defined above is denoted by $T_{P_{\mathit{NSL}_{\xor}}}$.
It is not hard to verify that we have
$T_{P_{\mathit{NSL}_{\xor}}}\svdash m(b,a)$ for every
$a,b\in \honest$.  In fact, secrecy of the nonces sent by
an honest responder to an honest initiator is not
guaranteed by the protocol
\cite{ChevalierKuestersRusinowitchTuruani-LICS-2003}.

    \section{Dominated Derivations} \label{sec:dominated}

In Section~\ref{sec:reduction}, we show how to reduce the
deduction problem modulo XOR to the one without XOR for
$\xor$-linear Horn theories, introduced below. This
reduction allows us to reduce the problem of checking
secrecy for protocols that use XOR to the case of protocols
that do not use XOR. (The authentication problem will be
considered in Section~\ref{sec:authentication}.) The latter
problem can then be solved by tools that cannot deal with
XOR, such as ProVerif. The class of protocol and intruder
capabilities that we can handle this way is quite large: It
contains all protocol and intruder rules that are
$\xor$-linear.

In this section, we prove a proposition that will be the key to the
reduction. Before we can state the proposition, we need to introduce
$\xor$-linear Horn theories and some further terminology.

A term is \emph{$\xor$-linear} if for each of its subterms
of the form $t\xor s$, where $t$ and $s$ may be standard or
non-standard terms, it is true that $t$ or $s$ is ground.
In other words, if a term $t$ contains a subterm of the
form $t_1\xor \cdots \xor t_n$ with $n\ge 2$, $t_i$
standard for every $i$, and there exists $i$ and $j$,
$i\not= j$, such that $t_i$ and $t_j$ are not ground, then
$t$ is not $\xor$-linear. For example, for variables
$x,y,z$ and a constant $a$, the term $t^1_{ex}=\an{a,a\xor
  \an{x,y}}$ is $\xor$-linear, but the term
$t^2_{ex}=\an{a,a\xor \an{x,y}\xor z}$ is not. A Horn
clause is called $\xor$-linear if each term occurring in
the clause is $\xor$-linear. A Horn theory is $\xor$-linear
if each clause in this theory, except for the $\xor$-rule
(see Fig.~\ref{intruder-rules}), is $\xor$-linear. In
particular, given a protocol $P$, the induced theory $T_P$
is $\xor$-linear if the sets of protocol and intruder
rules, except for the $\xor$-rule, are.

Our running example is an example of a protocol with an $\xor$-linear
Horn theory (note that, in \eqref{ckrt-fragile} and \eqref{ckrt-last},
$b$ is a constant); other examples are mentioned in
Section~\ref{sec:experiments}.  Also, many intruder rules are
$\xor$-linear. In particular, all those that do not contain the XOR
symbol. For example, in addition to the cryptographic primitives
mentioned in Figure~\ref{intruder-rules}, other primitives, such as
various kinds of signatures, encryption with prefix properties, and
MACs have $\xor$-linear intruder rules.

Besides $\xor$-linearity, we also need a more fine-grained
notion: $\C$-domination.  Let $\C$ be a finite set of
standard $\xor$-reduced ground terms such that $\C$ does
not contain two elements $m,m'$ with $m\not=m'$ and $m\sim
m'$.  (For the efficiency of our reduction
(Section~\ref{sec:reduction}), it is important to keep $\C$
as small as possible.) Let $\CC = \{ t \mid $ there exist
$c_1,\dots,c_n\in\C$ such that $t \sim c_1 \xor\cdots\xor
c_n \}$ be the $\xor$-closure of $\C$.  Note that
$0\in\CC$. Finally, let $\Csim = \{ t \mid t\sim t'\in\C,
t\mbox{ standard}\}$.
    
Now, a term is \emph{$\C$-dominated} if, for each of its
subterms of the form $t \xor s$, where $t$ and $s$ may be
standard or non-standard, it is true that $t$ or $s$ is in
$\CC$. For example, the term $t^1_{ex}$ from above is
$\{a\}$-dominated, but is is not $\{b\}$-dominated. The
term $t^2_{ex}$ is not $\{a\}$-dominated. A Horn clause is
$\C$-dominated, if the terms occurring in this clause are
$\C$-dominated; similarly for derivations. Finally, a Horn
theory $T$ is $\C$-dominated if each clause in $T$, except
for the $\xor$-rule, is $\C$-dominated. For example, we
have that the Horn theory $T_{P_{\mathit{NSL}_{\xor}}}$ of
our running example is $\{a,b\}$-dominated. (Recall that
$\part=\{a,b\}$.)

$\C$-dominated terms can also be characterized in terms of
what we call bad terms. We call a non-standard term $t$
\emph{bad} (w.r.t.~$\C$), if $t \sim c\xor t_1\xor
\dots\xor t_n$ for $c\in\CC$, pairwise $\xor$-distinct
standard terms $t_1,\dots,t_n\notin\Csim$, and $n>1$, where
$t$ and $t'$ are \emph{$\xor$-distinct} if $t\not\sim t'$.
A non-standard term which is not bad is called \emph{good}.
The following lemma is easy to see:
\begin{lemma}\label{lem:cdominatedbadterm}
 An $\xor$-reduced term is $\C$-dominated iff it contains no
bad subterms.
\end{lemma}
There is an obvious
connection between $\xor$-linearity  and $\C$-domination:

\begin{lemma}\label{lem:xorlineardominated}
  For every $\xor$-linear term/Horn theory/deriva\-tion there
  exists a finite set $\C$ of standard $\xor$-reduced
  messages such that the term/Horn theory/derivation is
  $\C$-dom\-inated.
\end{lemma}
The set $\C$ mentioned in the lemma could be chosen to be
the set of all ground standard terms occurring in the
term/Horn theory/derivation. However, $\C$ should be chosen
as small as possible in order to make the reduction
presented in Section~\ref{sec:reduction} more efficient.

As mentioned, the following proposition is the key to our
reduction. The proposition states that $\C$-dominated Horn
theories always allow for $\C$-dominated derivations.
Because of Lemma~\ref{lem:xorlineardominated},  the
proposition applies to all $\xor$-linear Horn theories.

\begin{proposition}\label{prop:good}
  Let $T$ be a $\C$-dominated Horn theory and $b$ be a
  $\C$-dominated fact.  If $T \svdash b$, then there
  exists a $\C$-dominated derivation modulo XOR for $b$
  from $T$.
\end{proposition}
Before we present the proof of this proposition, we
introduce some terminology, which is also used in
subsequent sections, and sketch the idea of the proof.
We write $t \ceq t'$ if $t' \sim c\xor t$ (or equivalently,
$c\xor t' \sim t$), for some $c\in\CC$.




For the rest of this section we fix a derivation $\pi$
modulo XOR for $b$ from $T$. W.l.o.g.~we may assume that
each term occurring in $\pi$ is in $\xor$-reduced form and
that each term in a substitution applied in $\pi$ is in
$\xor$-reduced form as well.

The key definitions for the proof of Proposition~\ref{prop:good} are
the following ones: 
\begin{definition}\label{def:type}
  For a standard term $t$, the set $\C$, and the derivation
  $\pi$, we define the \emph{type of $t$ (w.r.t. $\pi$ and
  $\C$)}, written $\type t$, to be an $\xor$-reduced
  element $c$ of $\CC$ such that $\pi(i) \sim \I(c\xor t)$ for some
  $i$, and for each $j<i$, it is not true that $\pi(j)\sim \I(c'\xor
  t)$ for some $c'\in\CC$.  If such an $i$ does not exist, we say that
  the type of $t$ is undefined.
\end{definition}
Note that the type of a term is uniquely determined
modulo AC and that equivalent terms (w.r.t.~$\sim$) have equivalent
types.

In the following definition, we define an operator which
replaces standard terms in bad terms which are not in
$\Csim$ by their types. This turns a bad term into a good
one. To define the operator, we use the following notation.
We write $\varphi_{\xor}[x_1,\ldots,x_n]$ for a term which
is built only from $\xor$, elements of $\Csim$, and the pairwise
distinct variables $x_1,\ldots,x_n$ such that each $x_i$ occurs
exactly once in $\varphi_{\xor}[x_1,\ldots,x_n]$. An example is
$\varphi^{ex}_{\xor}[x_1,x_2,x_3]=((x_1\xor x_2)\xor (a
\xor x_3))$, where $a\in\Csim$. For messages $t_1,\ldots,t_n$, we write
$\varphi_{\xor}[t_1,\ldots,t_n]$ for the message obtained from
$\varphi_{\xor}[x_1,\ldots,x_n]$ by replacing every $x_i$ by $t_i$,
for every $i\in \{1,\ldots,n\}$. Note that each non-standard term can
be expressed in the form $\varphi_\xor[t_1,\dots,t_n]$ for some
$\varphi_\xor$ as above and  standard terms $t_1,\dots,t_n \notin\Csim$.

\begin{definition}
  For a message $t$, we define $\Delta(t)$ as follows: If
  $t$ is a bad term of the form
  $\varphi_{\xor}[t_1,\ldots,t_n]$ for some
  $\varphi_{\xor}$ as above and standard terms
  $t_1,\dots,t_n \notin\Csim$, then $\Delta(t) =
  \varphi_{\xor}[\type t_1,\ldots,\type t_n]$; $\Delta(t)$
  is undefined, if one of those $\type t_i$ is undefined.
  Otherwise (if $t$ is good), we recursively apply $\Delta$
  to all direct subterms of~$t$.
\end{definition}

\noindent 
We will see (Lemma~\ref{lem:firstsimple}) that if $t$ occurs in $\pi$,
then the types of $t_i$ in the above definition are always
defined. Note also that $\Delta$ is defined with respect to the given
$\pi$ and $\C$.

Now, the main idea behind the proof of
Proposition~\ref{prop:good} is to apply $\Delta(\cdot)$ to
$\pi$. We then show that (i) $\Delta(\pi)$ is an incomplete
$\C$-dominated derivation modulo XOR for $b$ from $T$ and
(ii) to obtain a complete derivation only $\C$-dominated
terms are needed. The details of the proof are presented
next, by a series of lemmas, some of which are also used in
Section~\ref{sec:reduction}.

\parag{Proof of Proposition \ref{prop:good}}
The following lemma is easy to show by structural induction
on $s$: 
\begin{lemma}\label{lem:xorreducedbadterm}
 Let $s$ and $t$ be messages such that $s$ is $\xor$-reduced,
 $s$ contains a complete bad subterm $s'$, and $s\sim
 t$. Then, there exists a complete bad subterm $t'$ of $t$
 such that $t'\sim s'$. 
\end{lemma}

The following lemma, whose proof can be found in the appendix,
says that when substituting variables in a $\C$-dominated term, then
complete bad terms that might have been introduced by the substitution
cannot be canceled out by the $\C$-dominated term. 

\begin{lemma}\label{lem:presceq}
    Let $r\theta \sim t$, for a term $t$, an $\xor$-reduced
    substitution $\theta$, and a $\C$-dominated term $r$. Then, for
    each complete bad subterm $r'$ of $r\theta$ there exists a
    complete (bad) subterm $t'$ of $t$ such that $t'\sim r'$.
\end{lemma}

We now show (see the appendix) that if an instance of a
$\C$-dominated term contains a complete bad subterm, then this term
(up to $\ceq$) must be part of the substitution with which the
instance was obtained.

\begin{lemma}\label{lem:badtermsubstitution}
  Let $\theta$ be a ground substitution and $s$ be a $\C$-dominated
  term. Assume that $t$ is a complete bad subterm of $s\theta$. Then,
  there exists a variable $x$ and a complete bad subterm $t'$ of
  $\theta(x)$ such that $t'\ceq t$.
\end{lemma}

The converse of Lemma~\ref{lem:badtermsubstitution} is also
easy to show by structural induction on $s$.

 \begin{lemma}\label{lem:badtermsubstitutionconverse}
  Let $\theta$ be a ground substitution and $s$ be a
  $\C$-dominated term. If $s\theta$ is $\C$-dominated, then
  so is $\theta(x)$ for every $x\in \var(s)$. 
 \end{lemma}

Similarly to Lemma~\ref{lem:badtermsubstitution}, we can
prove the following lemma. The main observation is that
$\Delta(c\xor t) \sim c \xor \Delta(t)$, for $c\in\CC$.

\begin{lemma} \label{lem:linrule} 
  $\Delta(s\theta) \sim s(\Delta\theta)$, for a $\C$-dominated term
  $s$ and a substitution $\theta$.
\end{lemma}

Another basic and simple to prove property of $\Delta$ is
captured in the following lemma. 

\begin{lemma}\label{lem:equivalentdelta}
 Let $s$ and $t$ be terms such that $s\sim t$. Then,
 $\Delta(s)\sim \Delta(t)$. 
\end{lemma}

The following lemma says that if an instance of a
$\C$-dominated Horn clause contains a complete bad subterm
on its right-hand side, then this term (up to $\ceq$)
already occurs on the left-hand side.
\begin{lemma} \label{lem:badcopying} Assume that $p_1(r_1),
  \dots, p_n(r_n) \ra p_0(s)$ is a $\C$-dominated Horn
  clause, $\theta$ is an $\xor$-reduced ground
  substitution, $w, u_1, \dots, u_n$ are $\xor$-reduced
  messages such that $w \sim s\theta$ and $u_i \sim
  r_i\theta$, for $\rang i1n$.

        If $w'$ is a complete bad subterm of $w$, then there exists a
        complete bad subterm $u'$ of $u_i$, for some
        $\rang i1n$, such that $u'\ceq w'$.
    \end{lemma}
    \begin{proof}
      Suppose that $w'$ is a complete bad subterm of $w$.
      Because $w \sim s\theta$ and $w$ is $\xor$-reduced,
      by Lemma~\ref{lem:xorreducedbadterm}, there exists a
      complete bad subterm $t$ of $s\theta$ with $w'\sim
      t$.  By Lemma~\ref{lem:badtermsubstitution}, there
      exists a variable $x\in \var(s)$ and a complete bad
      subterm $t'$ of $\theta(x)$ with $t' \ceq t$. Because
      $x$, as a variable of $s$, has to occur also in $r_i$
      for some $\rang i1n$, the term $t'$ is a (not
      necessarily complete) subterm of $r_i\theta$. Since
      $r_i$ is $\C$-dominated, there exists a complete
      subterm $r'$ of $r_i\theta$ with $r'\ceq t'$.  Now,
      recall that $t'\ceq t$ and $t \sim w'$. It follows
      that $r'\ceq w'$.  Furthermore, since $w'$ is bad, so
      is $r'$.  Now, by Lemma~\ref{lem:presceq}, there
      exists a complete bad subterm $u'$ of $u_i$ such that
      $u'\ceq r' \ceq w'$.
    \end{proof}

    The following lemma connects bad terms that occur in a derivation
    with the types of their subterms.

    \begin{lemma} \label{lem:firstsimple} For every $n\ge
      1$, if $\pi(i) \sim \I(c\xor t_1\xxor t_n)$, for
      $c\in\CC$ and pairwise $\xor$-distinct standard terms
      $t_1,\dots,t_n\notin\Csim$, then, for each $k \in
      \{1,\dots,n\}$, there exists $j\le i$ such that
      $\pi(j)\sim \I(\type t_k \xor t_k)$.
    \end{lemma}
    \begin{proof}
        If $n=1$, then $\I(\type t_1 \xor t_1)$
        belongs to $\pi_{\leq i}$, by the definition of types.

        Now, suppose that $n>1$. In that case we will show,
        by induction on $i$, something more than what is
        claimed in the lemma: If $t$ with $t\sim c\xor
        t_1\xxor t_n$, $c\in \CC$, and pairwise
        $\xor$-distinct standard terms $t_i\notin \Csim$,
        occurs as a complete bad subterm in $\pi(i)$, then,
        for each $k \in \{1,\dots,n\}$, there exists $j\le
        i$ such that $\pi(j)\sim \I(\type t_k \xor t_k)$.

        Suppose that $t$, as above, occurs as a complete
        bad subterm in $\pi(i)$.

        If there exists $t'$ such that $t' \ceq t$ and $t'$
        occurs in $\pi_{<i}$ as a complete subterm, then we
        are trivially done by the induction hypothesis.
        (Note that $t'$ is bad since $t$ is.)  So, suppose
        that such a $t'$ does not occur in $\pi_{<i}$ as a
        complete subterm.  By Lemma \ref{lem:badcopying},
        $\pi(i)$ cannot be obtained by a $\C$-dominated
        Horn clause.  Thus, $\pi(i)$ is obtained by the
        $\xor$-rule, which means that $\pi(i)=\I(u)$ with
        $u \sim s \xor r$ for some $\I(s)$ and $\I(r)$
        occurring in $\pi_{<i}$.  We may assume that $s\sim
        d\xor s_1\xxor s_p$, with $d\in \CC$, and pairwise
        $\xor$-distinct $\xor$-reduced standard terms
        $s_1,\ldots,s_p\notin \Csim$, and $r\sim e\xor
        r_1\xxor r_q$, with $e\in \CC$, and pairwise
        $\xor$-distinct $\xor$-reduced standard terms
        $r_1,\ldots,r_q\notin \Csim$.

        According to our assumption, neither $s$ nor $r$
        contains a complete subterm $t'$ with $t' \ceq t$.
        In particular, neither $s$ nor $r$ contains $t'$
        with $t'\sim t$. So, since $\pi(i) \sim \I(s\xor
        r)$ contains $t$ as a complete subterm, it
        must be the case that $t \sim s\xor r$. Now, with
        $t\sim c\xor t_1\xor\dots\xor t_n$, as above, and
        $k\in \{1,\ldots,n\}$ it follows that either $s_l
        \sim t_k$ or $r_l \sim t_k$, for some $l$.  Suppose
        that the former case holds (the argument is similar
        for the latter case).  If $p>1$ (and thus $s$ is a
        bad term), then, by the induction hypothesis, we
        know that there exists $j<i$ such that $\pi(j)\sim
        \I(\type s_l\xor s_l)$. Since $t_k\sim s_l$, we
        have that $\type t_k\sim \type s_l$, and hence,
        $\pi(j)\sim \I(\type t_k \xor t_k)$.  Otherwise, $s
        \sim d\xor t_k$, and hence, by the definition of
        types, there exists $j<i$ with $\pi(j)\sim \I(\type
        t_k \xor t_k)$.
    \end{proof}

    The following lemma is the key in proving that
    $\Delta(\pi)$ is an incomplete derivation modulo XOR.

    \begin{lemma}\label{lem:cutcut} For every $i\le |\pi|$,
      if $\I(c\xor t_1\xxor t_n)$, for some $c\in\CC$ and
      pairwise $\xor$-distinct standard terms $t_1,\dots,t_n\notin
      \Csim$, belongs to $\pi_{<i}$, then there is a
      derivation for $\I(c\xor \type t_1\xxor \type t_n)$
      from $\Delta(\pi_{<i})$ modulo XOR.
    \end{lemma}
    \begin{proof}
      If $n=0$ or $n>1$, then $\I(c\xor \type t_1 \xxor
      \type t_n)\sim \I(\Delta(c\xor t_1\xxor t_n))$ by the
      definition of $\Delta$, and hence, $\I(c\xor \type
      t_1 \xxor \type t_n)$ can be derived from
      $\Delta(\pi_{<i})$.  So suppose that $n=1$. Since we
      have $\I(c \xor t_1)$ in $\pi_{<i}$, then, by the
      definition of types, we also have $\I(\type t_1\xor
      t_1)$ in $\pi_{<i}$. Thus, by the definition of
      $\Delta$, $\I(c \xor \Delta(t_1))$ and $\I(\type
      t_1\xor \Delta(t_1))$ are in $\Delta(\pi_{<i})$.
      From these one obtains $\I(c \xor \type t_1)$ by
      applying the $\xor$-rule.
    \end{proof}

    \medskip Now, we can finish the proof of
    Proposition~\ref{prop:good}. First, note that every
    non-standard message in $\Delta(\pi)$ is
    $\C$-dominated. This immediately follows from the
    definition of $\Delta$. We will now show (*): For each
    $i\in \{1,\ldots, |\pi|\}$, $\Delta(\pi(i))$ can be
    derived from $\Delta(\pi_{<i})$ modulo XOR by using
    only $\C$-dominated terms. This then completes the
    proof of Proposition~\ref{prop:good}. 

    Recall that we assume that $\pi$ is $\xor$-reduced and
    that in this derivation we use only $\xor$-reduced
    substitutions.  To prove (*), we consider two cases:
    \label{proof:good}

    \medskip

    \sloppypar
    \noindent \emph{Case 1.}  $\pi(i)$ is obtained from $\pi_{<i}$
    using a $\C$-dominated Horn clause $R = (p_1(s_1), \dots, p_n(s_n)
    \ra p_0(s_0))$ of $T$: Then there exists a $\xor$-reduced
    substitution $\theta$ such that $\pi(i) \sim p_0(s_0\theta)$ and
    the atoms $p_1(s_1\theta), \dots, p_n(s_n\theta)$ occur in
    $\pi_{<i}$ modulo XOR.  Thus, by Lemma~\ref{lem:equivalentdelta},
    $p_1(\Delta(s_1\theta)), \dots, p_n(\Delta(s_n\theta))$ occur in
    $\Delta(\pi_{<i})$ modulo XOR. Now, by Lemma~\ref{lem:linrule}, we
    have that $\Delta(s_i\theta)\sim s_i(\Delta\theta)$, for every
    $i\in \{0,\ldots,n\}$. Thus, by applying $R$ with the substitution
    $\Delta(\theta)$, we obtain $\Delta(\pi(i))\sim
    \Delta(s_0\theta)\sim s_0(\Delta(\theta))$.

    \medskip

    \noindent \emph{Case 2.}  $\pi(i)$ is obtained by the
    $\xor$-rule: Hence, there are two atoms $\I(s)$ and
    $\I(r)$ in $\pi_{<i}$ such that $\pi(i)\sim \I(s\xor
    r)$. We may assume that $s\sim c\xor s_1\xxor s_m$,
    with $c\in \CC$, and pairwise $\xor$-distinct
    $\xor$-reduced standard terms $s_1,\ldots,s_m\notin
    \Csim$, and $r\sim d\xor r_1\xxor r_l$, with $d\in
    \CC$, and pairwise $\xor$-distinct $\xor$-reduced
    standard terms $r_1,\ldots,r_l\notin \Csim$. Let
    $\{t_1,\dots,t_n\} = (S\setminus R) \cup (R\setminus
    S)$, for $S = \{s_1,\dots,s_m\}$ and $R =
    \{r_1,\dots,r_l\}$. Then, $\pi(i)\sim \I(s\xor r)\sim
    I(c\xor d\xor t_1\xxor t_n)$.  By Lemma
    \ref{lem:cutcut}, we know that $\I(c \xor \type s_1
    \xxor\type s_m)$ and $\I(d \xor \type r_1 \xxor\type
    r_l)$ can be derived from $\Delta(\pi_{<i})$ modulo
    XOR. Hence, $\I(t')$ with $t' = c\xor d \xor \type t_1
    \xxor \type t_n$ can be derived from $\Delta(\pi_{<i})$
    as well (by applying the $\xor$-rule).  Now, let us
    consider two cases:

        \begin{aenum}
        \item $n=0$ or $n>1$: \label{good:subcase} In this
          case, we have that $\Delta(\pi(i)) \sim \I(t')$,
          and hence, $\Delta(\pi(i))$ can be derived from
          $\Delta(\pi_{<i})$.

        \item $n=1$: Because $\I(c\xor s_1\xxor s_m)$ and
          $\I(d\xor r_1\xxor r_l)$ occur in $\pi_{<i}$
          modulo XOR, by Lemma~\ref{lem:firstsimple},
          $\I(\type t_1\xor t_1)$ occurs in $\pi_{<i}$
          modulo XOR as well. Thus, by
          Lemma~\ref{lem:equivalentdelta}, $\I(\type t_1
          \xor \Delta(t_1))$ occurs in $\Delta(\pi_{<i})$
          modulo XOR.  Now, because $\I(t')$, with $t'=
          c\xor d\xor \type t_1$, can be derived from
          $\Delta(\pi_{<i})$ modulo XOR, so can $\I(c\xor
          d\xor\Delta(t_1)) \sim \Delta(\pi(i))$.  \qed
        \end{aenum}

\section{The Reduction}\label{sec:reduction}

In this section, we show how the deduction problem modulo
XOR can be reduced to the deduction problem without XOR for
$\C$-dominated theories. More precisely, for a
$\C$-dominated theory $T$, we show how to effectively
construct a Horn theory $T^+$ such that a ($\C$-dominated)
fact can be derived from $T$ modulo XOR iff it can be
derived from $T^+$ in a syntactic derivation, where XOR is
considered to be a function symbol without any algebraic
properties. As mentioned, the syntactic deduction problem,
and hence, the problem of checking secrecy for
cryptographic protocols w.r.t.~an unbounded number of
sessions, can then be solved by tools, such as ProVerif,
which cannot deal with the algebraic properties of XOR.

In the remainder of this section, let $T$ be a
$\C$-dominated theory.  In what follows, we will first
define the reduction function, which turns $T$ into $T^+$,
and state the main result
(Section~\ref{sec:reductionfunction}), namely that the
reduction is sound and complete as stated above. Before proving this result
in Section~\ref{sec:reductionproof}, we illustrate the
reduction function by our running example
(Section~\ref{sec:reductionexample}).

\subsection{The Reduction
  Function}\label{sec:reductionfunction}

The reduction function uses an operator $\rep \cdot$, which
turns terms into what we call normal form, and a set
$\FSub(t)$ of substitutions associated with the term $t$.
We first define this operator and the set $\FSub(t)$.
The operator $\rep \cdot$ is defined w.r.t.~a linear
ordering $\precc$ on $\C$, which we fix once and for all.

\begin{definition}
For a $\C$-dominated term $t$, we define the \emph{normal form} of
$t$, denoted by $\rep{t}$, recursively as follows:
    \begin{compactitem}
    \item
        If $t$ is a variable, then $\rep t = t$.

    \item
        If $t=f(t_1,\dots,t_n)$ is standard, then
        $\rep{t} = f(\rep{t_1},\dots,\rep{t_n})$.

    \item
        If $t\in\CC$ is non-standard and $t \sim c_1\xor\cdots\xor c_n$, for
        some pairwise $\xor$-distinct $c_1,\dots,c_n\in\C$,
        $n > 1$, such that $c_1 \precc \cdots \precc c_n$,  then
        $\rep{t} = \rep{c_1} \xor (\rep{c_2}\xor(\cdots\xor
        \rep{c_n})\cdots)$.  

    \item
        If $t$ is non-standard and $t \sim c\xor t'$, 
        for some $c\in\CC$, $c\not\sim 0$, and standard $t'$ not in $\Csim$,
        then $\rep{t} = \rep{c} \xor \rep{t'}$.
    \end{compactitem}
    We say that a term $t$ is in \emph{normal form}, if
    $t=\rep t$. A substitution $\theta$ is in normal form,
    if $\theta(x)$ is in normal form for each variable $x$
    in the domain of $\theta$.
\end{definition}
It is easy to see that $\rep t = \rep s$ for $\C$-dominated
terms $t$ and $s$ iff $t\sim s$, and that $\rep t$ is
$\xor$-reduced for any $t$.  By $\CCnorm$, we denote the
set $\{\rep c \mid c\in\CC \}$.  Clearly, this set is
finite and computable in exponential time in the size of
$\C$.

To define the set $\FSub(t)$ of substitutions, we need the
notion of fragile subterms. For a $\C$-dominated term $t$,
the set of \emph{fragile subterms of $t$}, denoted by
$\frag(t)$, is $ \frag(t) = \{ s \mid $ $s$ is a
non-ground, standard term which occurs as a subterm of $t$ in
the form $t'\xor s$ or $s\xor t'$ for some $t'$\}. For
example, $\frag((a\xor \an{x,b})\xor b)=\{\an{x,b}\}$.

We are now ready to define the (finite and effectively
computable) set $\FSub(t)$ of substitutions for a
$\C$-dominated term $t$.  The main property of this set is
the following: For every $\C$-dominated, ground
substitution $\theta$ in normal form, there exists a
substitution $\sigma\in \FSub(t)$ and a substitution
$\theta'$ such that $\rep{t\theta}=(\rep{
  t\sigma})\theta'$. In other words, the substitutions in
$\FSub(t)$ yield all relevant instances of $t$. All ground,
normalized instances are syntactic instances of those
instances.  This resembles the finite variant property of
XOR \cite{ComonDelaune-RTA-2005} mentioned in the
introduction.  However, our construction of $\FSub(t)$ is
tailored and optimized towards $\C$-dominated terms and
substitutions.  More importantly, we obtain a stronger
property in the sense that the
equality---$\rep{t\theta}=(\rep{ t\sigma})\theta'$--- is
\emph{syntactic} equality, not only equality modulo AC; the
notion of $\C$-domination, which we introduced here, is
crucial in order to obtain this property.  Having syntactic
equality is important for our reduction in order to get rid
of algebraic properties completely.

\begin{definition}\label{def:F}
    Let $t$ be a $\C$-dominated term.  We define a
    family of substitutions $\FSub(t)$ as follows. The domain of
    every substitution in $\FSub(t)$ is the set of all variables which
    occur in some $s\in\frag(t)$.  Now, $\sigma\in\Sigma$, if for
    each $x\in\dom(\sigma)$ one of the following cases holds:
    \begin{ienum}
        \item 
            $\sigma(x) = x$, 
        \item 
            $x\in\frag(t)$ and $\sigma(x)=c\xor x$, for some
            $c\in\CCnorm$, $c\neq 0$,
          \item there exists $s\in\frag(t)$ with
            $x\in\var(s)$ and a $\C$-dominated substitution
            $\theta$ in normal form such that
            $s\theta\in\CC$ and $\sigma(x)= \theta(x)$.
    \end{ienum}
\end{definition}
To illustrate the definition and the property mentioned
above, consider, as an example, $t=c\xor x$ and the
substitution $\theta(x)=d\xor m$, with $d\in \CCnorm$ and a
$\C$-dominated, standard term $m\notin \CCnorm$ in normal
form. In this case, we can choose $\sigma(x)=d\xor x$
according to (ii). With $\theta'(x)=m$, we obtain
$\rep{t\theta}=\rep{c\xor d}\xor m=(\rep{t\sigma})\theta'$.
If $\theta(x)$ were $d\in \CCnorm$, then (iii) would be
applied.


    %
We can show (see the appendix):
\begin{lemma}\label{lem:sigmacomputable}
  For a $\C$-dominated term $t$, the set $\FSub(t)$ can be computed in
  exponential time in the size of $t$.
\end{lemma}
We are now ready to define the reduction function which turns $T$ into
$T^+$. The Horn theory $T^+$ is given in Fig.~\ref{fig:Tplus}. With
the results shown above, it is clear that $T^+$ can be constructed in
exponential time from $T$.
\begin{figure*}[!t] 
    \centering
    \begin{align}
        \label{rule:prot}
        \rep{r_1\sigma}, \dots, \rep{r_n\sigma} &\ra \rep{r_0\sigma}
        &&  \parbox[t]{31em}{\small
                for each $\C$-dominated rule $r_1,\dots,r_n\ra r_0$ of
                $T$ and each $\sigma \in \FSub(\an{r_0,\dots,r_n})$.
            } \\[.2ex]
        \label{rule:const}
        \I(c), \I(c') &\ra \I(\rep{c\xor c'})
            &&\text{\small for each $c,c'\in\CCnorm$} \\[.2ex]
        \label{rule:pop}
        \I(c), \I(x) &\ra \I(c\xor x)
            &&\text{\small for each $c\in\CCnorm$} \\[.2ex]
        \label{rule:variant}
        \I(c), \I(c'\xor x) &\ra \I(\rep{c\xor c'}\xor x)
            &&\text{\small for each $c,c'\in\CCnorm$} \\[.2ex]
        \label{rule:gen}
        \I(c \xor x), \I(c' \xor x) &\ra \I(\rep{c\xor c'})
            &&\text{\small for each $c,c'\in\CCnorm$} 
    \end{align}
    \caption{\label{fig:Tplus} Rules of the theory $T^+$.
      We use the convention that $I(0\xor x)$ stands for
      $I(x)$. }
\end{figure*}
The Horn clauses in \eqref{rule:const}--\eqref{rule:gen}
simulate the $\xor$-rule in case the terms we consider are
$\C$-dominated.  The other rules in $T$ are
simulated by the rules in \eqref{rule:prot}, which are constructed in
such a way that they allow us to produce messages in normal form for
input messages in normal form.

We can now state the main theorem of this paper. This
theorem states that a message (a secret) can be derived
from $T$ using derivations modulo XOR if and only if it can
be derived from $T^+$ using only syntactic derivations,
i.e., no algebraic properties of XOR are taken into
account.  As mentioned, this allows to reduce the problem
of verifying secrecy for cryptographic protocols with XOR,
to the XOR-free case. The latter problem can then be
handled by tools, such as ProVerif, which otherwise could
not deal with XOR.

\begin{theorem}\label{the:reduction}
  For a $\C$-dominated Horn theory $T$ and $\C$-dominated
  message $b$ in normal form, we have: $T \svdash b$ if and only if $T^+
  \vdash b$.
\end{theorem}

Before we prove this theorem, we illustrate the reduction
by our running example.

\subsection{Example}\label{sec:reductionexample}

Consider the Horn theory $T_{P_{\mathit{NSL}_{\xor}}}$ of
our running example. As mentioned in
Section~\ref{sec:dominated}, this Horn theory is
$\C$-dominated for $\C=\{a,b\}$. In
what follows, we illustrate how
$T^+_{P_{\mathit{NSL}_{\xor}}}$ looks like, where the
elements of $\C$ are ordered as $a<_{\C} b$.

First, consider the instances of Horn clauses of
$T_{P_{\mathit{NSL}_{\xor}}} $ given by \eqref{rule:prot}.
Only the Horn clauses in \eqref{ckrt-fragile} have fragile
subterms. All other Horn clauses have only one instance in
$T^+_{P_{\mathit{NSL}_{\xor}}}$: the rule itself. This is
because for such Horn clauses $\FSub(\cdot)$ contains only
one substitution, the identity. The Horn clause in
\eqref{ckrt-fragile} has one fragile subterm, namely $x$.
Hence, the domain of every substitution in the
corresponding $\Sigma$-set is $\{x\}$, and according to
Definition \ref{def:F}, this set contains the following
eight substitutions: item (i) gives 
$\sigma_1 = \{x/x\}$;
item (ii) gives 
$\sigma_2 = \{a\xor x/x\}$, $\sigma_3 =
\{b\xor x/x\}$, and $\sigma_4 = \{(a\xor b)\xor x/x\}$;
item (iii) gives 
$\sigma_5 = \{0/x\}$, $\sigma_6 = \{a/x\}$,  
$\sigma_7 = \{b/x\}$, and $\sigma_8 = \{a\xor
b/x\}$. 
For each of these substitutions we obtain an instance of
\eqref{ckrt-fragile}.  For example, $\sigma_4$ yields
$$
    \I(\penc{\an{(a\xor b)\xor x, a}}{\pub{sk_b}}) 
        \ra \I(\penc{\an{m(b,a), a \xor x}}{\pub{sk_a}}) .
$$

Now, consider the Horn clauses induced by
\eqref{rule:const}--\eqref{rule:gen}.  For example, the set
of Horn clauses \eqref{rule:variant} contains among others:
$\I(a\xor b), \I(b\xor x) \ra \I(a \xor x) $ and 
$\I(b), \I(a\xor x) \ra \I((a\xor b) \xor x)$.

\subsection{Proof of Theorem~\ref{the:reduction}\label{sec:reductionproof}}

In what follows, let $T$ be a $\C$-dominated Horn theory
and $b$ be a $\C$-dominated message in normal form. Note
that $\rep b=b$.  The following lemma proves that our reduction is
sound, i.e., that $T^+ \vdash b$ implies $T \svdash b$.

\begin{lemma}\label{lem:toleft}
    If $\pi$ is a syntactic derivation for $b$ from $T^+$, then $\pi$ is a
    derivation for $b$ from $T$ modulo XOR. 
\end{lemma}
\begin{proof}
  Let $\pi$ be a syntactic derivation for $b$ from
  $T^+$. To prove the lemma it suffices to prove that each
  $\pi(i)$ can be obtained by a derivation modulo XOR from
  $T$ and $\pi_{<i}$.  If $\pi(i)$ is obtained from
  $\pi(j)$ and $\pi(k)$ for $j,k<i$, using one of the Horn clauses
  \eqref{rule:const}--\eqref{rule:gen}, then we can apply
  the $\xor$-rule with $\pi(j)$ and $\pi(k)$ to obtain
  $\pi(j)\xor\pi(i) \sim \pi(i)$.

  Now, suppose that $\pi(i)$ is obtained using a Horn
  clause in \eqref{rule:prot} of the form
  $\rep{r_1\sigma},\dots,\rep{r_n\sigma}\
  \ra\rep{r_0\sigma}$ for some Horn clause
  $(r_1,\dots,r_n\ra r_0)\in T$ and some $\sigma \in
  \FSub(\an{r_0,\dots,r_n})$.  Hence, there exists a
  substitution $\theta$ and, for each $\rang k1n$, there
  exists $j<i$ such that $\pi(j) = \rep{r_k\sigma}\theta
  \sim (r_k\sigma)\theta = r_k(\sigma\theta)$. So, we can
  use the rule $r_1,\dots,r_n\ra r_0$ to obtain
  $r_0(\sigma\theta) = (r_0\sigma)\theta \sim
  \rep{r_0\sigma}\theta = \pi(i)$.  Note that $\rep t \sim
  t$ and if $t\sim t'$, then $t\sigma \sim t'\sigma$ for
  all terms $t,t'$ and substitutions $\sigma$.
\end{proof}
To prove the completeness of our reduction, i.e., that $T
\svdash b$ implies $T^+ \vdash b$, we first prove the
property of $\FSub(t)$ mentioned before
Definition~\ref{def:F}. For this, we need the following
definition.

\begin{definition}\label{def:sigma}
    Let $t$ be a $\C$-dominated term and $\theta$ be a
    $\C$-dominated, ground substitution in normal form with
    $\dom(\theta)=\var(t)$.  Let $\sigma=\sigma(t,\theta)$ be
    the substitution defined as follows.  The domain of
    $\sigma$ is the set of all variables that
    occur in some $s\in\frag(t)$. Let $x$ be such a
    variable. We define $\sigma(x)$ according to the following
    conditions, which have decreasing priority: 
    \begin{aenum}
        \item  If there exists $s\in\frag(t)$ with
          $x\in\var(s)$ such that $s\theta \in \CC$, then $\sigma(x)=\theta(x)$.

        \item Otherwise, if $x\in\frag(t)$ and $\theta(x) =
          c \xor s'$, for $c\in\CC$ and some standard term
          $s'$ not in $\CC$, then $\sigma(x) = c\xor x$.
          (Note that $c\not=0$ since $\theta(x)$ is in
          normal form.)

        \item  Otherwise, $\sigma(x)=x$. (Note that in this
          case we know that $\theta(x)$ is some standard
          term not in $\CC$ if $x\in \frag(t)$.)
    \end{aenum}
\end{definition}
Equipped with this definition, we show (see the appendix) the property
of $\FSub(t)$ mentioned before Definition~\ref{def:F}. 
\begin{lemma}\label{lem:propertySigmat}
 Let $t$ be a $\C$-dominated term and $\theta$ be a
    $\C$-dominated, ground substitution in normal form with
    $\dom(\theta)=\var(t)$. Then,
    $\sigma=\sigma(t,\theta)\in \FSub(t)$ and there exists
    a substitution $\theta'$ such that
    $\theta=\sigma\theta'$, i.e.,
    $\theta(x)=\sigma(x)\theta'$ for every $x\in
    \dom(\theta)$,  and
    $\rep{t'\theta}=\rep{t'\sigma}\theta'$ for every
    subterm $t'$ of $t$. 
\end{lemma}

We can now show the completeness of our reduction.

\begin{lemma} \label{lem:toright}
    If $\pi$ is a $\C$-dominated derivation for $b$ from
    $T$ modulo XOR, then $\rep\pi$ is a syntactic
    derivation for $b$ from $T^+$. 
\end{lemma}
\begin{proof}
  We show that every $\rep{\pi(i)}$ can be derived
  syntactically from $T^+$ and $\rep{\pi_{<i}}$.  Two cases
  are distinguished:

  \medskip\noindent \textbf{Case 1:} $\pi(i)$ is obtained
  from $\pi(j) = \I(t)$ and $\pi(k)=\I(s)$, for $j,k<i$,
  using the $\xor$-rule.  In that case $\pi(i) \sim
  \I(t\xor s)$. By assumption $t$, $s$, and $t\xor s$ are
  $\C$-dominated, and hence, $\rep t$, $\rep s$,
  $\rep{t\xor s}$ are either normalized standard terms not
  in $\CC$, terms in $\CCnorm$, or terms of the form $c\xor
  u$ for $c\in \CCnorm$ and a normalized standard term
  $u\notin \CC$, respectively. However, it is not the case
  that $\rep t=c\xor u$ or $\rep t=u$ and $\rep s=u'\notin
  \CC$ or $\rep s=c'\xor u'$ with $u\not=u'$ since
  otherwise $\rep{t\xor s}$ would not be $\C$-dominated.
  Now, it is easy to see that $\xor$-rule can be simulated
  by one of the Horn clauses
  \eqref{rule:const}--\eqref{rule:gen}.

  \medskip\noindent \textbf{Case 2:} $\pi(i)$ is obtained
  using some $\C$-dominated rule $(r_1,\dots,r_n \ra
  r_0)\in T$ and a ground substitution $\theta$.  Since
  $\pi$ is $\C$-dominated, by
  Lemma~\ref{lem:badtermsubstitutionconverse} and
  \ref{lem:xorreducedbadterm} we may assume that $\theta$
  is $\C$-dominated. Since $\pi$ is a derivation modulo
  XOR, we may also assume that $\theta$ is in normal form.
  We have that $\pi(i) \sim r_0\theta$ and there exist
  $j_1,\dots, j_n<i$ such that $\pi(j_k) \sim r_k\theta$,
  for all $\rang k1n$.
        
  Let $\sigma = \sigma(\an{r_0,\dots,r_n},\theta)$ and let
  $\theta'$ be as specified in
  Lemma~\ref{lem:propertySigmat}. By
  Lemma~\ref{lem:propertySigmat}, $\sigma\in
  \FSub(\an{r_0,\dots,r_n})$.  Now, to obtain
  $\rep{\pi(i)}$, we can use the rule $\rho =
  (\rep{r_1\sigma}, \dots, \rep{r_n\sigma} \ra
  \rep{r_0\sigma})\in T^+$ with the substitution $\theta'$.
  In fact, by Lemma~\ref{lem:propertySigmat}, we have that
  $\rep{r_k\sigma}\theta' = \rep{r_k\theta}=\rep{\pi(j_k)}$
  for all $\rang k0n$, where $j_0=0$. (Recall that for
  $\C$-dominated terms $s$ and $t$ with $s\sim t$, we have
  that $\rep s=\rep t$.)
\end{proof}
Now, from the above lemma and Proposition~\ref{prop:good}
it immediately follows that $T \svdash b$ implies $T^+
\vdash b$.

    \section{Authentication}\label{sec:authentication}

In the previous section, we showed how to reduce the
derivation problem modulo XOR for $\C$-dominated Horn
theories to the syntactic derivation problem. While the
derivation problem corresponds to the secrecy problem for
cryptographic protocols w.r.t.~an unbounded number of
sessions, in this section, we will see that it is not hard
to extend our result to authentication properties.

\subsection*{Authentication as Correspondence Assertions}

Authentication properties are often expressed as
\emph{correspondence assertions} of the form $\eEnd(x) \ra
\eBegin(x)$ where $x$ describes the parameters on which the
begin and end events should agree. This correspondence
should be read as follows: If event $\eEnd(x)$ has
occurred, then also event $\eBegin(x)$. For example,
$\eEnd(a,b,n) \ra \eBegin(a,b,n)$ could be interpreted as:
If $b$ thinks to have finished a run of a protocol with $a$
in which the nonce $n$ was used (in this case event
$\eEnd(a,b,n)$ occurred), then $a$ has actually run a
protocol with $b$ in which $n$ was used (in this case event
$\eBegin(a,b,n)$ occurred).  To check such correspondence assertions
in the Horn theory based approach, roughly speaking, the protocol
rules are augmented with atoms representing events of the form
$\eBegin(x)$ and $\eEnd(x)$ (see, e.g., \cite{Blanchet-archive-2008}
for details).

\begin{figure*}[t!]
     \begin{align} 
 \label{eq:sidfirst}        &\I(\enc{\pub{k_b}}{n(a,b,\sid), a})
             &\text{for every $a\in\honest$, $b\in\part$} \\[1ex]
 \label{eq:sidsecond}        \I(\enc{\pub{k_b}}{x, a}) \ \ra \ &\I(\enc{\pub{k_a}}{m(b,a,\sid,x), x\xor b})
             &\text{for every $b\in\honest$, $a\in\part$} \\[1ex]
 \label{eq:beginevent}        \eBegin(a,b,y), \
         \I(\enc{\pub{k_a}}{y, n(a,b,\sid)\xor b}) \ \ra\  &\I(\enc{\pub{k_b}}{y})
             &\text{for every $a\in\honest$, $b\in\part$} \\[1ex]
         \I(\enc{\pub{k_b}}{(x,a)}), \
 \label{eq:endevent}        \I(\enc{\pub{k_b}}{m(b,a,\sid,x)})\  \ra \ &\eEnd(a,b,m(b,a,\sid,x))
             &\text{for every $b\in\honest$, $a\in\part$} 
     \end{align}
     \caption{Rules for authentication ($\sid$ is a variable intended to
     range over session identifiers). \label{fig:ckrt-auth}}
\end{figure*}

For our running example, this is illustrated in
Figure~\ref{fig:ckrt-auth}. In \eqref{eq:endevent}, the end
event indicates that $b$ believes to have talked to $a$ and
the nonce $m(b,a,sid,x)$ was used in the interaction, where
$x$ is the nonce $b$ believes to have received from $a$ and
$sid$ is a session identifier.  The parameters $x$ and
$\sid$ are added to the term representing the nonce in
order to make the analysis more precise. In particular, the
session identifier is added in order to make the
correspondence stronger: The events should not only
correspond on the names and the nonces used in the protocol
run, but also on the session identifiers.  Note that
without the session identifier, correspondence of sessions
would otherwise not be guaranteed since in the Horn theory
based approach new protocol runs do not necessarily use
completely fresh nonces. The begin event in
\eqref{eq:beginevent} indicates that $a$ just received the
response from $b$ and now outputs her response to $b$,
where the begin event contains the nonce received from $b$.

We note that, strictly speaking, the Horn theory depicted
in Figure~\ref{fig:ckrt-auth} falls out of the class of
Horn theories that we allow, not because of
$\xor$-linearity but because of the fact that the variable
$sid$ occurs on the right-hand side of a Horn clause but
not on the left-hand side (see \eqref{eq:sidfirst} and
\eqref{eq:sidsecond}).  However, as we noted in Section
\ref{sec:preliminaries}, this assumption can easily be
relaxed for variables that are supposed to be substituted
only by $\C$-dominated terms, which is the case for session
identifiers.

Now, let $T$ be a Horn theory model of a protocol and an
intruder, i.e., $T$ consists of a set of protocol rules
(such as those in Figure~\ref{fig:ckrt-auth}), a set of
initial intruder facts, and a set of intruder rules.
Following Blanchet \cite{Blanchet-archive-2008}, we say
that a (non-injective) correspondence assertion of the form
$\eEnd(x) \ra \eBegin(x)$ is satisfied by $T$ if
 \begin{equation}\label{auth}
\parbox{.82\columnwidth}{ for every finite set of messages
  $B$ and every message $m_0 \notin \widehat{B}$, it holds that $T
  \cup \{\eBegin(m)\mid m\in B\}\mbox{$\not\!\!\svdash$}
  \eEnd(m_0)$,  }
\end{equation}
where $\widehat{B}=\{t\mid$ there exists $t'\in
B$ and $t\sim t'\}$. In \cite{Blanchet-archive-2008}, this
formulation (more precisely, a syntactic version, i.e., the
XOR-free version) is somewhat implicit in a theorem which
reduces correspondence assertions in process calculus to
Horn theories. Blanchet then proposes a method for proving
the syntactic version of \eqref{auth} using ProVerif.


\subsection*{Extending Our Reduction to Correspondence
Assertions}\label{sec:reductioncorrespondences}

The following theorem extends our reduction presented in
Section~\ref{sec:reduction} to the problem of solving
\eqref{auth} \emph{with} XOR. In fact, we show that if in
\eqref{auth} the ($\C$-dominated) Horn theory $T$ is
replaced by $T^+$ (i.e., we can use the same reduction
function as in Section~\ref{sec:reduction}), then
derivation modulo XOR ($\svdash$) can be replaced by
syntactic derivation ($\;\vdash\;$). Now, the latter problem
(the syntactic version of \eqref{auth}) can be solved using
ProVerif. Formally, we can prove:

\begin{theorem}\label{the:authentication}
  Let $T$ be a $\C$-dominated Horn theory.
  Then, \eqref{auth} holds iff for every
  finite set of messages $B$ and every message $m_0\notin
  B$, it holds that $T^+\cup\{\eBegin(m)\mid m\in
  B\}\mbox{$\not\,\vdash$} \ \eEnd(m_0)$.
\end{theorem}
The proof of this theorem requires some slight extension of
Proposition~\ref{prop:good}, stated below, in which an injective
version of $\Delta$ is used, i.e., $t\not\sim t'$ should imply that
$\Delta(t)\not\sim\Delta(t')$. This is needed to guarantee that if
$m_0\notin \widehat B$, then $\Delta(m_0)\notin \widehat{\Delta(B)}$.


This can be achieved by fixing an \emph{injective} function $\gamma$ which
takes a term to some term built from $0$ and $\an{\cdot,\cdot}$ (or
any other function which the intruder can apply). We also add the
fresh constant $\cc$ to the intruders knowledge. Now, for a bad term
$t=c\xor t_1\xor \cdots \xor t_n$, we define $\Delta(t)=c\xor
\type{t_1}\xor \cdots \xor \type{t_n}\xor \enc\cc{\gamma(t)}$. The
important property of $\enc\cc{\gamma(t)}$ is that the intruder can
derive this message and that it is unique for every term $t$.

\begin{proposition}\label{prop:authgood}
   Let $T$ be a $\C$-dominated Horn theory, 
   $B$ be a finite set of facts, and $a$ be a fact. If $T \cup B
   \svdash a$, then there exists a $\C$-dominated derivation for
   $\Delta(a)$ from $T \cup \Delta(B)$ modulo XOR.
\end{proposition}
The proof of this proposition is very similar to the one of
Proposition~\ref{prop:good}. Only minor modifications are
necessary.



Now, to prove Theorem \ref{the:authentication}, it suffices
to show that the following conditions are equivalent, for a
$\C$-dominated theory $T$:
\begin{ienum}
\item 
    there exist a finite set of messages $B$ and a message $m_0\notin
    \widehat B$ such that $T \cup \{\eBegin(m)\mid m\in B\} \  \svdash \
    \eEnd(m_0)$ 
\item
    there exist a finite set of $\C$-dominated messages $B$ and a
    $\C$-dominated message $m_0\notin \widehat B$ such that
    $T \cup\{\eBegin(m)\mid m\in B\} \ \svdash \
    \eEnd(m_0)$.
\item
    there exist a finite set of $\C$-dominated messages $B$ and a
    $\C$-dominated message $m_0\notin B$ such that
    $T^+\cup\{\eBegin(m)\mid m\in B\} \ \vdash \
    \eEnd(m_0)$.
\item
    there exist a finite set of messages $B$ and a message $m_0\notin
    B$ such that $T^+\cup\{\eBegin(m)\mid m\in B\} \ \vdash \
    \eEnd(m_0)$.
\end{ienum}
\begin{proof}
    The implication (i)$\Rightarrow$(ii) follows from 
    Proposition~\ref{prop:authgood} and by the fact that $\Delta$ is
    injective; (ii)$\Rightarrow$(iii) is given by
    Theorem~\ref{the:reduction} (we use the fact that $T \cup
    \{\eBegin(m)\mid m\in B\}$ is $\C$-dominated and the fact that
    $(T \cup \{\eBegin(m)\mid m\in B\})^+ = T^+ \cup
    \{\eBegin(m)\mid m\in \rep B\} $  );  (iii)$\Rightarrow$(iv) is
    trivial; finally, (iv)$\Rightarrow$(i)  is given by
    Lemma~\ref{lem:toleft}.
\end{proof}


\section{Implementation and Experiments}\label{sec:implementationandexperiments}

We have implemented our reduction, and together with
ProVerif, tested it on a set of protocols which employ the
XOR operator (see \cite{implementation} for the implementation). In
this section, we report on our implementation and the experimental
results.

\subsection{Implementation}\label{sec:implementation}

We have implemented our reduction function in SWI prolog
(version 5.6.14).  Our implementation essentially takes a
Horn theory as input. More precisely, the input consists of
(1) a declaration of all the functor symbols used in the
protocol and by the intruder, (2) the initial intruder
facts as well as the protocol and intruder rules, except
for the $\xor$-rule, which is assumed implicitly, (3) a
statement which defines a secrecy or authentication goal.
Moreover, options that are handed over to ProVerif may be
added.

Our implementation then first checks whether the given Horn
theory, say $T$, (part (2) of the input) is $\xor$-linear.
If it is not, an error message is returned. If it is, a set
$\C$ is computed such that the Horn theory is
$\C$-dominated. Recall that such a set always exists if the
Horn theory is $\xor$-linear. It is important to keep $\C$
as small as possible, in order for the reduction to be more
efficient. Once $\C$ is computed, the reduction
function as described in Section~\ref{sec:reduction} is
applied to $T$, i.e., $T^+$ is computed. Now, $T^+$
together with the rest of the original input is passed on
to ProVerif. This tool then does the rest of the work,
i.e., it checks the goals for $T^+$. This is possible
since, due the reduction, the XOR operator in $T^+$ can now
be considered to be an operator without any algebraic
properties.

Our implementation does not follow the construction of the
reduction function described in Section~\ref{sec:reduction}
precisely, in order to produce an output that is optimized
for ProVerif (but still equivalent): a) While terms of the
form $c\xor t$, with $c\in \CC$, $t\notin \CC$ are 
represented by $\tt xor(c,t)$, terms $a\xor b\in \CCnorm$
are represented by $\tt xx(a,b)$. This representation
prevents some unnecessary unifications between terms.
However, it is easy to see that with this representation,
the proofs of soundness and completeness of our reduction
still go through. The basic reason is that terms in
$\CCnorm$ can be seen as constants. b) For the Horn clauses
in Figure~\ref{fig:Tplus},
\eqref{rule:const}--\eqref{rule:gen}, we do not produce
copies for every choice of $c,c'\in \CCnorm$. Instead, we
use a more compact representation by introducing auxiliary
predicate symbols. For example, the family of Horn clauses
in \eqref{rule:variant} is represented as follows:
$\tt{xtab}(x,y,z), I(y), I(\tt{xor}(x,t)) \rightarrow
I(\tt{xor}(z,t))$, where the facts
$\tt{xtab}(c,c',\rep{c\xor c'})$ for every $c,c'\in
\CCnorm$ are added to the Horn theory given to ProVerif.

\subsection{Experiments}\label{sec:experiments}

\begin{figure}[!t]
    \centering
    \begin{tabular}{lcrr}
        \sf\footnotesize protocol & \sf\footnotesize correct & \sf\footnotesize reduction time & \sf\footnotesize ProVerif time \\
        \hline\\[-1.2ex]
        \sf NSL$_{\xor}$        & no  & 0.02s & 0.006s \\
        \sf NSL$_{\xor}$-fix    & yes & 0.04s & 0.09s \\
        \sf NSL$_{\xor}$-auth-A & no  & 0.03s & 0.16s \\
        \sf NSL$_{\xor}$-auth-A-fix & yes  & 0.03s & 0.02s \\
        \sf NSL$_{\xor}$-auth-B & yes  & 0.04s & 0.5s \\
        \sf SK3     & yes & 0.05s & 0.3s \\
        \sf RA      & no  & 0.05s & 0.17s \\
        \sf RA-fix  & yes & 0.05s & 0.27s \\
        \sf CCA-0         & no  & 0.15s & 109s \\
        \sf CCA-1A        & yes & 0.06s & 0.7s \\
        \sf CCA-1B        & yes & 0.07s & 1.3s \\
        \sf CCA-2B        & yes & 0.14s & 7.1s \\
        \sf CCA-2C        & yes & 0.15s & 58.0s \\
        \sf CCA-2E        & yes & 0.07s & 1.42s 
    \end{tabular}
    \caption{Experimental Results. \label{fig:exp}} 
\end{figure}

We applied our method to a set of ($\xor$-linear)
protocols. The results, obtained by running our
implementation on a 2,4 Ghz Intel CoreTM 2 Duo E6700
processor with 2GB RAM, are depicted in
Figure~\ref{fig:exp}, where we list both the time of the
reduction and the time ProVerif needed for the analysis of
the output of the reduction. We note that except for
certain versions of the CCA protocol, the other protocols
listed in Figure~\ref{fig:exp} are out of the scope of the
implementation in \cite{CortierKeighrenSteel-TACAS-2007},
the only other implementation that we know of for
cryptographic protocol analysis w.r.t.~an unbounded number
of sessions that takes XOR into account. As mentioned in
the introduction, the method in
\cite{CortierKeighrenSteel-TACAS-2007} is especially
tailored to the CCA protocol. It can only deal with
symmetric encryption and the XOR operator, but, for
example, cannot deal with protocols that use public-key
encryption or pairing. Let us discuss the protocols and
settings that we analyzed in more detail.

By \mbox{NSL$_{\xor}$} we denote our running example. Since
there is an attack on this protocol, we also propose a fix
\mbox{\sf NSL$_{\xor}$-fix} in which the message
$\penc{\an{M, N\xor B}}{\pub{sk_A}}$ is replaced by
$\penc{\an{M, h(\an{N,M})\xor B}}{\pub{sk_A}}$ for a hash
function $h(\cdot)$. We analyze both authentication and
secrecy properties for these ($\xor$-linear) protocols.

The ($\xor$-linear) protocol {\sf SK3}
\cite{ShoupRubin-EUROCRYPT-1996} is a key distribution
protocol for smart cards, which uses the XOR operator. {\sf
  RA} denotes an ($\xor$-linear) group protocol for key
distribution \cite{BullOtway-TR-DRA-1997}. Since there is a
known attack on this protocol, we proposed a fix: a message
$k_{A,B} \xor h(\an{\mathsf{key}(A),N})$ sent by the key
distribution server to $A$ is replaced by $k_{A,B}\xor
h(\an{\mathsf{key}(A),\an{N,B}})$.

\begin{figure*}[!t]
         \begin{align}
            I(x),\ I(\e{k}{\km\xor\data})  \rra  
                I(\e{x}{k}) \tag{{\it Encipher}} \label{cca:enc}\\[1ex]
            I(\e x k),\ I(\e k{\km\xor\data} )
                \rra I(x)  \tag{{\it Decipher}} \label{cca:dec} \\[1ex]
            I(\e{k}{\km\xor\typ}),\ I(\typ),\ I(\e{\kek}{\km\xor\exp})
                \rra I(\e{k}{\kek\xor\typ})  \tag{{\it KeyExport}} \label{cca:exp} \\[1ex]
            I(\e{k}{\kek \xor \typ}),\ I(\typ),\ I(\e{\kek}{\km\xor\imp})
                \rra I(\e{k}{\km\xor\typ})  \tag{{\it KeyImport}} \label{cca:imp} \\[1ex]
            I(k1), I(\typ) \rra I(\e{k1}{\km\xor\kp\xor\typ})
                \tag{{\it KeyPartImp-First}} \label{cca:first} \\[1ex]
            I(k2), I(\e{x}{\km\xor\kp\xor\typ}), I(\typ)
                \rra I(\e{x\xor k2}{\km\xor\kp\xor\typ})
                \tag{{\it KeyPartImp-Middle}} \label{cca:middle} \\[1ex]
            I(k3), I(\e{y}{\km\xor\kp\xor\typ}), I(\typ)
                \rra I(\e{y\xor k3}{\km\xor\typ})
                \tag{{\it KeyPartImp-Last}} \label{cca:last} \\[1ex]
            I(\e{k}{\kek_1\xor\typ}),\ I(\typ),\ I(\e{kek_1}{\km\xor\imp}),\
                I(\e{kek_2}{\km\xor\exp})  \rra  I(\e{k}{kek_2\xor\typ})
                \tag{{\it KeyTranslate}} \label{cca:transl}
         \end{align}
        \caption{CCA API \label{fig:cca}, where $\km$ denotes a constant
          (the key master stored in the cryptographic coprocessor),
          $\typ$ is a constant that ranges over the constants in
          $\{\data,\imp,\exp,\pin\}$, 
          and all other symbols
          ($x$, $y$, $k$, ...) are variables.} 
\end{figure*}


{\sf CCA} stands for Common Cryptographic Architecture
(CCA) API \cite{IBM-CCA} as implemented on the hardware
security module IBM 4758 (an IBM cryptographic
coprocessor). The CCA API is used in ATMs and mainframe
computers of many banks to carry out PIN verification
requests. It accepts a set of commands, which can be seen
as receive-send-actions, and hence, as cryptographic
protocols. The only key stored in the security module is
the master key $\km$. All other keys are kept outside of
the module in the form $\e{k}{\km\xor \typ}$, where
$\typ\in \{\data,\imp,\exp,\pin\}$ denotes the type of the
key, where each type is some fixed constant.  The commands
of the CCA API include the following: Commands for
encrypting/decrypting data using data keys. Commands to
export/import a key to/from another security module. This
is done by encrypting/decrypting the key by a
key-encryption-key. 

In Figure~\ref{fig:cca}, we model the most important
commands of the CCA API (see also
\cite{CortierKeighrenSteel-TACAS-2007}) in terms of Horn
clauses.  \eqref{cca:enc} and \eqref{cca:dec} are used to
encrypt/decrypt data by data keys.  \eqref{cca:exp} is used
to export a key to another security module by encrypting it
under a key-encryption-key, with \eqref{cca:imp} being the
corresponding import command. The problem is to make the
same key-encryption-key available in different security
modules. This is done by a secret sharing scheme using the
commands \eqref{cca:first}--\eqref{cca:last}, where $\kp$
is a type (a constant) which stands for ``key part'',
$\kek$ is obtained as $k1\xor k2 \xor k3$, and each $ki$,
$i\in \{1,2,3,\}$, is supposed to be known by only one
individual.  \eqref{cca:transl} is used to encrypt a key
under a different key-encryption-key.


We note that some of the Horn clauses in Figure~\ref{fig:cca}, namely
\eqref{cca:middle} and \eqref{cca:last}, are not linear.
Fortunately, one can apply a standard unfolding technique for Horn
clauses together with straightforward simplifications to obtain an
\emph{equivalent} Horn theory with only $\xor$-linear rules.

There are several known attacks on the CCA API, which
concern the key-part-import process. One attack is by Bond
\cite{Bond-CHES-2001}. As a result of this attack the intruder is able
to obtain PINs for each account number by performing
data encryption on the security module. A stronger attack
was found by IBM and is presented in \cite{Clulow-Msc}
where the intruder can obtain a PIN derivation key, and
hence, can obtain PINs even without interacting with the
security module. However, the IBM attack depends on key
conjuring \cite{CortierKeighrenSteel-TACAS-2007}, and
hence, is harder to carry out. Using our implementation
(together with ProVerif) and the configuration denoted by
{\sf CCA-0} in Figure~\ref{fig:exp}, we found a new attack
which achieves the same as the IBM attack, but is more
efficient as it does not depend on key conjuring. Our attack is
presented at the end of this section.

In response to the attacks reported in \cite{Bond-CHES-2001}, IBM
proposed two recommendations.

\medskip

\noindent\emph{Recommendation 1.}  As mentioned, the
attacks exploit problems in the key-part-import process. To
prevent these problems, one IBM recommendation is to
replace this part by a public-key setting. However, as
shown in \cite{CortierKeighrenSteel-TACAS-2007}, further
access control mechanisms are needed, which essentially
restrict the kind of commands certain roles may perform.
Two cases, which correspond to two different roles, are
considered, and are denoted {\sf CCA-1A} and {\sf CCA-1B}
in Figure~\ref{fig:exp}. We note that the Horn theories
that correspond to these cases are $\xor$-linear, and
hence, our tool can be applied directly, no changes are
necessary; not even the transformations mentioned above.
Since public-key encryption (and pairing) cannot be
directly handled by the tool presented by Cortier et
al.~\cite{CortierKeighrenSteel-TACAS-2007}, Cortier et
al.~had to modify the protocol in an ad hoc way, which is
not guaranteed to yield an equivalent protocol. This is
also why the runtimes of the tools cannot be compared
directly.

\medskip

\noindent\emph{Recommendation 2.} 
Here additional access control mechanisms are assumed which
ensure that no single role is able to mount an attack.  We
analyzed exactly the same subsets of commands as the ones
in \cite{CortierKeighrenSteel-TACAS-2007}. These cases are
denoted {\sf CCA-2B}, {\sf -2C}, and {\sf -2E} in
Figure~\ref{fig:exp}, following the notation in
\cite{CortierKeighrenSteel-TACAS-2007}. The runtimes
obtained in \cite{CortierKeighrenSteel-TACAS-2007} are
comparable to ours: 333s for {\sf CCA-2B}, 58s for {\sf
  -2C}, and 0.03s for {\sf -2E}.

\subsubsection*{Our Attack.}\label{sect:attack}

As we noted before, our tool found an attack
which---according to our knowledge---has not been
discovered before.  This attack uses the same assumptions
as Bond's attack in terms of the role played by the
intruder and his knowledge. As in the IBM attack, we use
the fact that $0$ is the default value for $\data$.
  
Our attack does not use key conjuring, and hence, is easier
to carry out than the IBM attack. As a result of the
attack, the intruder obtains a pin derivation key in clear
(like in the IBM attack).

\medskip In the attack we assume that a new
key-encryption-key $\kek$ needs to be imported, using the
three-part key import commands
\eqref{cca:first}--\eqref{cca:last}, which means that $\kek
= k1\xor k2\xor k3$, where $k1$, $k2$, $k3$ are the shares
known by three different individuals.
  
The key $\kek$ is then used to import a new pin-derivation
key $\pdk$ to the security module, in the form
\begin{equation}\label{pdk}
       \e{\pdk}{\kek\xor\pin}.
\end{equation}
We assume that this message can be seen by the attacker and
that the attacker is the third participant
of the process of importing $\kek$. In particular, the
attacker can perform \eqref{cca:last}, knows the value $k3$, and obtains the
message
\begin{equation}\label{kp2}
        \e{k1\xor k2}{\km\xor\kp\xor\imp}.
\end{equation}

Now we describe the steps of the attack.  After the
intruder receives \eqref{kp2}, he uses \eqref{cca:last}
with $k3\xor\pin$ instead of $k3$. In this way he obtains
\begin{equation}\tag{A1}\label{at1}
        \e{\kek \xor \pin}{\km\xor\imp}
\end{equation}
He uses the same command again, this time with
$k3\xor\pin\xor\exp$, obtaining:
\begin{equation}\tag{A2}\label{at2}
        \e{\kek \xor \pin \xor \exp}{\km\xor\imp}
\end{equation}
Next, when $\pdk$ is imported, the intruder uses \eqref{cca:imp}
twice: The first time with input \eqref{at1}, \eqref{pdk}, and $\typ =
\data = 0$, resulting in the message
\begin{equation}\tag{A3}\label{at3}
        \e{\pdk}{\km\xor\data}.
\end{equation}
The second time with input \eqref{at2}, \eqref{pdk},
and $\typ =
\exp$, resulting in the message
\begin{equation}\tag{A4}\label{at4}
        \e{\pdk}{\km\xor\exp}.
\end{equation}
Now, using \eqref{cca:exp} with input \eqref{at3},
\eqref{at4}, and $\typ = \data = 0$, the attacker obtains
\begin{equation}\tag{A5}\label{at5}
        \e{\pdk}{\pdk\xor\data} = \e\pdk\pdk.
\end{equation}
Finally, using \eqref{cca:dec} with input \eqref{at5} and
\eqref{at3}, the attacker obtains the clear value of
$\pdk$, which can be used to obtain the PIN for any account
number: Given an account number, the corresponding PIN is
derived by encrypting the account number under $\pdk$.


    \appendix
    \section{Proofs for Section \ref{sec:dominated}}

In what follows we will use the following notation: $t\topAC t'$ if
$t$ and $t'$ are coincide up to transformation modulo AC, with
standard terms kept unchanged. For example, $(a\xor \an{a\xor
b,b})\xor b\topAC (a\xor b) \xor \an{a\xor b,b} \not\!\!\topAC (a\xor
b ) \xor \an{b\xor a,b}$.

    \subsection*{Proof of Lemma \ref{lem:presceq}.}
      Assume that $r'$ is a complete bad subterm of
      $r\theta$.  We proceed by structural induction on $r$
      and consider the following cases:
        \begin{itemize}
        \item $r = x$ is a variable: Because $\theta$ is
          $\xor$-reduced, so is $\theta(x)$. So, since $r'$
          is a subterm of $\theta(x)$ and $\theta(x)\sim
          t$, Lemma~\ref{lem:xorreducedbadterm} implies
          that there exists a complete bad subterm $t'$ of
          $t$ with $t'\sim r'$.

        \item $r = f(r_1,\dots,r_n)$, for $f \neq \xor$: In
          this case, $t$ is of the form $f(t_1,\dots,t_n)$
          with $t_i\sim r_i\theta$.  Since $r\theta$ is not
          bad, $r'$ is a subterm of $r_i\theta$ for some
          $i\in\{1,\dots,n\}$. By the induction hypothesis,
          there exists a complete bad subterm $t'$ of $t_i$
          (and thus, of $t$) with $t'\sim r'$.

        \item $r=c$, for $c\in\CC$: We have that
          $r\theta=r$. Since $r$ is $\C$-dominated it
          follows that $c$ does not contain complete bad
          subterms. Hence, nothing is to show.

        \item $r\topAC c\xor r''$
          with $c\in\CC$ and $r''\notin \CC$ standard, but
          not a variable: The case that $r'=r\theta$ cannot
          occur since this term is not a bad term.  Since
          $r$ is $\C$-dominated, $c$ does not contain a
          complete bad subterm. Hence, $r'$ cannot be a
          subterm of $c\theta=c$.  So $r'$ is a subterm of
          $r''\theta$.

            Let $s\sim
            r''\theta$, for some $\xor$-reduced term $s\in
            \CC$. So, we have that $t \sim c\xor s$.  Since
            $r''$, as a proper subterm of $r$, is
            $\C$-dominated, from the fact that $r'$ is a
            complete bad subterm of $r''\theta$ it follows
            by the induction hypothesis that there exists a
            complete bad subterm $t'$ of $s$ with $r'\sim
            t'$. Now, since $c$ is $\C$-dominated (because
            by assumption $r$ is), and hence, $c$ does not
            contain complete bad subterms, it follows that
            $t'$ occurs as a subterm in $t$.


          \item $r\topAC c\xor x$,
            for $c\in\CC$ and a variable $x$: Assume that
            $\theta(x) \sim c'\xor t_1\xor\cdots\xor t_n$
            with $n\ge 0$, $c'\in \CC$, and pairwise
            $\xor$-distinct standard terms
            $t_1,\dots,t_n\notin\Csim$. First assume that
            $r' = r\theta$, which implies that $n>1$.  Then
            we can set $t'=t$ since $t'=t\sim r\theta=r'$.
            Otherwise, since $r$ is $\C$-dominated, it
            follows that $c$ does not contain a
            complete bad subterm. Hence, $r'$ is a complete
            bad subterm of $c'$ or there exists $i$ such
            that $r'$ is a complete bad subterm of $t_i$.
            In any case, this term, let us call it $t''$,
            does not coincide with any standard term $c_i$
            with $c=c_1\xor \ldots \xor c_k$ because these
            terms do not contain complete bad subterms.
            Hence, $t''$ is equivalent to some term $t'$ in
            $t$.  Thus, there exists a complete bad subterm
            $t'$ of $t$ with $r'\sim t'$.
            \qed
        \end{itemize}

\subsection*{Proof of Lemma \ref{lem:badtermsubstitution}.}
  We proceed by structural induction on $s$:
        \begin{itemize}
        \item \emph{$s=x$ is a variable}: We can set $t' = t$.
        \item \emph{$s$ is standard}: Then $s\neq t$, and
          thus, for one of the direct subterms $s'$ of $s$,
          $s'\theta$ has to contain $t$ as a complete
          subterm. By the induction hypothesis, there
          exists a variable $x\in\var(s')\subseteq \var(s)$
          such that $\theta(x)$ contains a complete bad
          subterm $t'$ with $t'\ceq t$.
        \item \emph{$s\in \CC$}: This case is not possible,
          since $s=s\theta$ is $\C$-dominated, and hence,
          cannot contain a complete bad subterm.
        \item \emph{$s\topAC c\xor s'$, where
            $c\in\CC$ and $s'\notin \CC$ is standard, but not a
            variable}: Then, $t\not=s\theta$ since
          $s\theta$ is not a bad term. Moreover, $c$ is
          $C$-dominated (since it belongs to $s$), and
          hence, cannot have $t$ as a subterm. Hence, $t$
          must be a subterm of $s'\theta$ and we can use
          the induction hypothesis.
        \item
            \emph{$s \topAC c\xor x$, for $c\in\CC$ and a
            variable $x$}: If $t\sim (c\xor x)\theta$, we can
          choose $t'=\theta(x)$, since $t'\ceq
          t$. Otherwise, since $c$ is $\CC$-dominated, and
          hence, does not contain complete bad subterms, it
          follows that $t$ is a subterm of
          $\theta(x)$. Hence, we can choose $t'=t$. 
          \qed
        \end{itemize}

\section{Proofs for Section \ref{sec:reduction}}

\subsection*{Proof of Lemma \ref{lem:sigmacomputable}.}

We start with showing that matching of $\C$-dominated terms modulo XOR
yields a uniquely determined matcher modulo XOR, if any, and this
matcher can be computed in polynomial time.

\medskip\noindent\emph{Claim 1.}
  Let $s$ be a $\C$-dominated term and $t$ be a ground term. Then, the
  matcher of $s$ against $t$ is uniquely determined modulo XOR, i.e.,
  if $s\theta\sim t$ and $s\theta'\sim t$ for substitutions $\theta$
  and $\theta'$, then $\theta(x)\sim \theta'(x)$ for every $x\in
  \var(s)$. Moreover, the matcher of $s$ against $t$ can be
  computed in polynomial time in the size of $s$ and $t$.


\begin{proof}
    We show how to compute the unique (modulo XOR) matcher of
    $s$ against $t$. The computed matcher will be in normal form.
    First, for substitutions $\sigma_1$ and $\sigma_2$ we define
    $\sigma_1\sqcup\sigma_2$ as $\sigma_1\cup\sigma_2$ if for each
    $x\in \dom(\sigma_1)\cap\dom(\sigma_2)$ we have that
    $\sigma_1(x)=\sigma_2(x)$. Otherwise, $\sigma_1\sqcup\sigma_2$ is
    undefined.  
    
    We obtain the matcher $\sigma$ of $s$ against $t$ recursively as follows.
    We can assume that both $s$ and $t$ are in normal form (one can
    transform a term $t$ into its normal form $\rep t$ in polynomial
    time)\footnote{So far, we defined $\rep\cdot$ only for
    $\C$-dominated terms. Now, we need to extend the definition of 
    $\rep\cdot$ to work for all terms. Such a extension is
    straightforward. So we skip it.}.
    We consider the following cases:

    \begin{nenum}
    \item
        $s=x$ is a variable: Then $\sigma = \{t/x\}$.
    \item
        $s$ is a ground term:
        Then $\sigma = \emptyset$ if $s=t$.  Otherwise, the
        matcher does not exist.

    \item
        $s = c \xor s'$, for ground $c$ and nonground, standard $s'$:
        Then $\sigma$ is the matcher of $s'$ against the term 
        $\rep{c\xor t}$.

    \item
        $s=f(s_1,\dots,s_n)$, for $f\neq\xor$, non ground: 
        
        If $t=f(t_1,\dots,t_n)$, we
        take $\sigma = \sigma_1 \sqcup \cdots \sqcup \sigma_n$, where
        $\sigma_i$, for $\rang i1n$, is the matcher of $s_i$ against
        $t_i$. Otherwise, i.e.\ if such a $\sigma$ does not exist, the
        matcher does not exist.  
    \end{nenum}

    It is easy to show that this algorithm computes a
    matcher of $s$ against $t$, if it exists, and moreover, that this
    matcher is unique.
\end{proof}

Now, we are ready to prove Lemma \ref{lem:sigmacomputable}: The domain
of every substitution in $\FSub(t)$ is polynomial, since it is a
subset of $\var(t)$. Hence, it suffices to show that for every
variable in the domain there are only exponentially many possible
values and these values can be computed effectively. This is clear
for the case (i) and (ii) in Definition~\ref{def:F}, as $\CCnorm$ is
bounded exponentially (in the size of $\C$). 

As for case (iii), let $s, x$ and $\theta$ be given as in this case.
Note that $s$ is $\C$-dominated.  Hence, $\theta$ is the unique
matcher of $s$ against some $c\in\CCnorm$.  Because $\theta$ can be
computed from $s$ and $c$ in polynomial time and, moreover, both $s$
and $c$ range over exponentially bounded sets (in fact, $\frag(t)$ is
polynomial and $\CCnorm$ is exponential), the claim of the lemma
follows.

\subsection*{Proof of Lemma \ref{lem:propertySigmat}.}
  Let $t$ and $\theta$ be given as in the lemma. By
  construction, it is easy to see that
  $\sigma=\sigma(t,\theta)\in \FSub(t)$. It is also easy to
  see that there exists $\theta'$ such that
  $\theta=\sigma\theta'$ and the domain of $\theta'$ is the
  set of all variables that occur in some $\sigma(x)$ for
  $x\in \dom(x)$. Note that $\theta'$ is uniquely
  determined.  Let $t'$ be a subterm of $t$. We need to
  show that $\rep{t'\theta}=\rep{t'\sigma}\theta'$.  We
  proceed by structural induction on $t'$.

  First, suppose that $t'\in\var(t)$: Let $x=t'$. We distinguish the
  following cases: 
    \begin{aenum}
    \item If $\sigma(x)$ was defined according to
      Definition~\ref{def:sigma}, (a), then
      $\sigma(x)=\theta(x)$. It follows that
      $\rep{x\theta}=\rep{x\sigma}\theta'$. 
    \item Otherwise, if $\sigma(x)$ was defined according to
      Definition~\ref{def:sigma}, (b), then $x\in\frag(t)$,
      $\theta(x) = c \xor s'$, for $c\in\CCnorm$ and some
      normalized standard term $s'$ not in $\CC$, and
      $\sigma(x) = c\xor x$. It follows that $\theta'(x)=s'$
      and $\rep{x\sigma}\theta' = \rep{c\xor x}\theta' = (c\xor
      x)\theta' = c\xor s' = \rep{c\xor s'} = \rep{x\theta}$.
    \item Otherwise, if $\sigma(x)$ was defined according to
      Definition~\ref{def:sigma}, (c), then
      $\sigma(x)=x$ and $\theta'(x)=\theta(x)$. Since $\theta(x)$ is normalized, it
      follows that $\rep{x\theta}=\rep{x\sigma}\theta'$. 
    \end{aenum}

 Second, suppose that $t'= f(t_1,\dots,t_n)$,
    for $f\neq\xor$: By the induction hypothesis, it
    follows that $\rep{t'\theta}=
    f(\rep{t_1\theta},\ldots,\rep{t_n
      \theta})=f(\rep{t_1\sigma}\theta',\ldots,\rep{t_n\sigma}\theta')=\rep{t'\sigma}\theta'$. 
 
 If we suppose that $t'\sim c$, for $c\in \CCnorm$, then it immediately
 follows that $\rep{t'\theta}=\rep{t'\sigma}\theta'$. 

 Now, suppose that $t'\sim c\xor x$, for $c\in \CCnorm$: We distinguish the following cases: 
\begin{aenum}
\item If $\sigma(x)$ was defined according to
  Definition~\ref{def:sigma}, (a), then
  $\sigma(x)=\theta(x)$. It follows that
  $\rep{t'\theta}=\rep{t'\sigma}\theta'$. 
\item Otherwise, if $\sigma(x)$ was defined according to
  Definition~\ref{def:sigma}, (b), then $x\in\frag(t)$,
  $\theta(x) = c' \xor s'$, for $c'\in\CCnorm$ and some
  normalized standard term $s'$ not in $\CC$, and
  $\sigma(x) = c'\xor x$. It follows that $\theta'(x)=s'$
  and $\rep{t'\sigma}\theta' = \rep{c\xor c'\xor x}\theta'
  = \rep{c\xor c'}\xor x\theta' = \rep{c\xor c'}\xor s' =
  \rep{c\xor c'\xor s'} = \rep{t'\theta}$.
\item Otherwise, if $\sigma(x)$ was defined according to
  Definition~\ref{def:sigma}, (c), then $\sigma(x)=x$
  and $\theta'(x)=\theta(x)$. Since $x\in\frag(t)$ and
  items (a) and (b) of Definition~\ref{def:sigma} do not hold,
  $\theta'(x)$ is a normalized standard term not in $\CCnorm$. It
  follows that $\rep{t'\theta}=\rep{c\xor \theta(x)}=c\xor
  \theta(x)=\rep{t'\sigma}\theta'$.
\end{aenum}

Finally, suppose that
$t'\sim c\xor s$, for $c\in \CCnorm$ and a $\C$-dominated,
standard subterm $s$ of $t'$ with $s\notin\CC$ and
$s\notin\var(t)$: We distinguish the
following cases: 
\begin{aenum}
\item If $s\theta\in \CC$, then $\sigma(x)$, for $x\in
  \var(s)$, was defined
  according to Definition~\ref{def:sigma}, (a) since $s\in \frag(t)$.
  Hence, $\sigma(x)=\theta(x)$ for all $x\in \var(s)$, and thus
  $s\sigma$ is ground and $s\sigma=s\theta$.
  It follows that
  $\rep{t'\theta}=\rep{c\xor s\theta}=\rep{x\xor
    s\sigma}=\rep{x\xor s\sigma}\theta'=\rep{t'\sigma}\theta'$.
\item Otherwise, if $s\theta\notin \CC$, by the induction hypothesis
    it follows that $\rep{s\theta}=\rep{s\sigma}\theta'$. We have also
    that $s\sigma$ is not in $\CC$ (otherwise, $s\theta$ would be also
    in $\CC$).  Moreover, since $s\theta\notin \CC$, we obtain that
    $\rep{t'\theta}=c\xor \rep{s\theta}=c\xor
    \rep{s\sigma}\theta'=\rep{(c\xor
    s)\sigma}\theta'=\rep{t'\sigma}\theta'$. 
    \qed
\end{aenum}


    {\small 
    \bibliography{literature} }

\end{document}